%% file: ijcai19.tex
\documentclass{article}
\usepackage{ijcai19}

\usepackage{times}
\usepackage{soul}
\usepackage{url}
\usepackage[hidelinks]{hyperref} 	%DE-COMMENT!!!
\usepackage[utf8]{inputenc}
\usepackage[small]{caption}
\usepackage{graphicx}
\usepackage{amsmath}
\usepackage{booktabs}
\usepackage{algorithm}
\usepackage{algorithmic}
\usepackage{amssymb}
\usepackage{amsmath}
\usepackage{tikz}
\usetikzlibrary{calc}
\usepackage{todonotes}
\usepackage{dsfont}
\usepackage{enumitem}
\usepackage{tikz}
\usepackage{amsthm}
\usepackage{thm-restate}
\usepackage{multirow}
\usepackage{amsfonts}
\usepackage{mathtools}
\usepackage{wrapfig}
\usepackage{subcaption}
\newtheorem{theorem}{Theorem}
\newtheorem{corollary}{Corollary}[theorem]

\newtheorem{definition}{Definition}

\newtheorem{observation}{Observation}

\newtheorem{remark}{Remark}

\title{Leadership in Congestion Games:\\ Multiple User Classes and Non-Singleton Actions \\ (Extended Version)}

\author{
Alberto Marchesi$^1$\footnote{Equal Contribution}
\and
Matteo Castiglioni$^{1\ast}$
\And
Nicola Gatti$^1$
\affiliations
$^1$ Politecnico di Milano, Piazza Leonardo da Vinci 32, Milano, Italy 
\emails
\{alberto.marchesi, matteo.castiglioni, nicola.gatti\}@polimi.it
}
\begin{document}

\maketitle

\input{abstract}

\input{introduction}

\input{related_work}

\input{preliminaries}

\input{congestion_games}

\input{singleton_games}

\input{formulations}

\input{experiments}

\input{conclusions}

\section*{Acknowledgments}

This work has been partially supported by the Italian MIUR PRIN 2017 Project ALGADIMAR ``Algorithms, Games, and Digital Market''.

\clearpage

%
%
%\appendix
%
%

\bibliographystyle{named}
\bibliography{refs}

\vspace*{\fill}

\appendix
\input{appendix}

\end{document}

%% file: abstract.tex
\begin{abstract}  % put your abstract here!
We study the problem of finding \emph{Stackelberg equilibria} in games with a massive number of players. So far, the only known game instances in which the problem is solved in polynomial time are some particular \emph{congestion games}. However, a complete characterization of hard and easy instances is still lacking. In this paper, we extend the state of the art along two main directions. First, we focus on games where players' actions are made of multiple resources, and we prove that the problem is \textsf{NP}-hard and \emph{not} in Poly-\textsf{APX} unless \textsf{P} = \textsf{NP}, even in the basic case in which players are symmetric, their actions are made of only two resources, and the cost functions are monotonic.
Second, we focus on games with singleton actions where the players are partitioned into \emph{classes}, depending on which actions they have available.
In this case, we provide a dynamic programming algorithm that finds an equilibrium in polynomial time, when the number of classes is fixed and the leader plays pure strategies.
Moreover, we prove that, if we allow for leader's mixed strategies, then the problem becomes \textsf{NP}-hard even with only four classes and monotonic costs. Finally, for both settings, we provide mixed-integer linear programming formulations, and we experimentally evaluate their scalability on both random game instances and worst-case instances based on our hardness reductions.
\end{abstract}

%% file: introduction.tex
\section{Introduction}\label{sec:introduction}

In the last years, Stackelberg games and their corresponding \emph{Stackelberg equilibria} (SEs) received a lot of attention in the artificial intelligence literature, thanks to their successful applications in real-world settings, \emph{e.g.}, in security domains~\cite{tambe2011security}. 
In a Stackelberg game, there is a special player, called \emph{leader}, who has the ability to commit to a (potentially mixed) strategy beforehand, while the other players, who act as \emph{followers}, decide how to play after observing the commitment~\cite{von2010leadership}.
As most of the works in the literature, see, \emph{e.g.}, \cite{conitzer2006computing,paruchuri2008playing,conitzer2011commitment},
%---e.g., \citeauthor{conitzer2006computing}~\shortcite{conitzer2006computing}, \citeauthor{paruchuri2008playing}~\shortcite{paruchuri2008playing}, \citeauthor{conitzer2011commitment}~\shortcite{conitzer2011commitment}---, 
we focus on \emph{optimistic} SEs (OSEs), which assume that the followers break ties in favor of the leader, maximizing her utility.
Few works also analyze \emph{pessimistic} SEs, where the followers break ties so as to minimize the leader's utility.
However, the pessimistic assumption usually results in computational problems that are much harder than their optimistic counterparts, as showed, \emph{e.g.}, by~\citeauthor{basilico2017bilevel}~[\citeyear{basilico2017bilevel}] and~\citeauthor{coniglio2017pessimistic}~[\citeyear{coniglio2017pessimistic,algo2018computing}].

It is well-known that, while finding an OSE in two-player normal-form games is easy~\cite{conitzer2006computing}, the problem becomes much harder in games with multiple followers playing a \emph{Nash equilibrium} (NE).
%, even when the number of followers is fixed.
%
For instance, %when the followers play a \emph{Nash equilibrium} (NE) in the game resulting from the leader's commitment, 
\citeauthor{basilico2019bilevel}~[\citeyear{basilico2016methods,basilico2019bilevel}] show that computing an OSE is \textsf{NP}-hard and not in Poly-\textsf{APX} unless \textsf{P} = \textsf{NP} in normal-form games with only two followers, and \citeauthor{alberto2018computing}~[\citeyear{alberto2018computing}] show the same result for polymatrix games with an arbitrary number of followers.
Moreover, finding an OSE is \textsf{NP}-hard also in normal-form games with multiple followers playing sequentially~\cite{conitzer2006computing}, while the problem becomes easy if they play simultaneously in a correlated manner~\cite{conitzer2011commitment}.
%
%{\bf with three followers restricted to play pure strategies only}.
%

Recently, \citeauthor{Marchesi18:leadership}~[\citeyear{Marchesi18:leadership}] and \citeauthor{aij2018leadership}~[\citeyear{aij2018leadership}] studied the problem of computing OSEs in \emph{congestion games} (CGs).
In a CG, the players' actions are subsets of a given set of shared resources, and each resource has a cost that depends on the number of players using it (a.k.a. \emph{congestion}).
Most of the computational studies on CGs, see, \emph{e.g.}, \cite{fabrikant2004complexity,ieong2005fast,ackermann2008impact}, focus on finding pure-strategy NEs.
%, as they always exist in these games.
%
Instead, \citeauthor{Marchesi18:leadership}~[\citeyear{Marchesi18:leadership}] and \citeauthor{aij2018leadership}~[\citeyear{aij2018leadership}] apply the Stackelberg paradigm to CGs, assuming that the followers play a pure-strategy NE after observing the leader's mixed-strategy commitment.
These works focus on \emph{singleton} CGs, where actions are made of only one resource.
Specifically, \citeauthor{Marchesi18:leadership}~[\citeyear{Marchesi18:leadership}] show that an OSE can be found in polynomial time when costs are \emph{monotonic} in the congestion and players are \emph{symmetric} (\emph{i.e.}, they share the same set of actions), even when the number of followers is non-fixed.
\citeauthor{aij2018leadership}~[\citeyear{aij2018leadership}] show that the same holds even with non-monotonic costs if we restrict the leader to pure strategies.
%
%Thus, 
These results provide some examples of multi-follower Stackelberg games where finding an equilibrium is computationally tractable.
On the other hand, \citeauthor{aij2018leadership}~[\citeyear{aij2018leadership}] also show that, if either the costs are non-monotonic or players are not symmetric, then the problem becomes \textsf{NP}-hard and not in Poly-\textsf{APX}, unless~\textsf{P} = \textsf{NP}.

Identifying new classes of multi-follower Stackelberg games where SEs can be computed in time polynomial in the number of players is crucial for real-world problems, enabling the adoption of SEs in massive applications, such as, \emph{e.g.}, resource-sharing systems with premium (prioritized) users, cybersecuirty with multiple treats, and influence maximization in social networks.
Despite this, a complete characterization of hard and easy game instances is still lacking.

% However, a complete characterization of hard and easy game instances is still unknown and could play a crucial role in practical applications, enabling the adoption of Stackelberg games in massive scenarios.

\subsubsection{Original Contributions}

In this paper, we extend the state of the art on the computational problem of finding OSEs in CGs, identifying new cases in which it is solvable in polynomial time and others where it is not.
%
%We pursue our goal along two different directions.
%
First, we show that having actions made of only one resource is necessary to have efficient (polynomial-time) algorithms. 
Indeed, we prove that finding an OSE is \textsf{NP}-hard and not in Poly-\textsf{APX} unless \textsf{P} = \textsf{NP}, even if players' actions contain only two resources, costs are monotonic, and players are symmetric.
Then, we introduce and study singleton CGs in which the players are partitioned into \emph{classes}, with followers of the same class sharing the same set of actions.
These are a generalization of singleton CGs with symmetric players, capturing the common case in which users can be split into (usually few) different classes, such as, \emph{e.g.}, users with different priorities.
For these games, we provide a dynamic programming algorithm that computes an OSE in polynomial time, when the number of classes is fixed and the leader is restricted to play pure strategies.
%, thus extending the positive result presented by~\citeauthor{aij2018leadership}~[\citeyear{aij2018leadership}].
%
On the other hand, we prove that, if the leader is allowed to play mixed strategies, then the problem becomes \textsf{NP}-hard even with only four classes and monotonic costs.
Finally, for both settings, we design \emph{mixed-integer linear programming} (MILP) formulations for computing OSEs, and we experimentally evaluate them on a testbed containing both randomly generated game instances and worst-case instances based on our hardness reductions.~\footnote{All the proofs of Lemmas and Theorems are in Appendix~\ref{sec:app_a}.}
%~\footnote{For reasons of space, we omit all the proofs of Lemmas and Theorems. We provide full proofs and additional experimental results in the online Appendices~\ref{sec:app_a}~and~\ref{sec:app_b}, respectively, which the reader can access at \url{https://drive.google.com/file/d/1ob2Mv8GhTw6nL-f9J-QoTozkGdfGAFot/view?usp=sharing}.}
%
%, showing that they provide good performance in practice. 

%% file: related_work.tex
\section{Related works}
\label{sec:related_work}

In this work, we apply the Stackelberg paradigm to CGs following the approach of \citeauthor{von2010leadership}~[\citeyear{von2010leadership}] and \citeauthor{conitzer2006computing}~[\citeyear{conitzer2006computing}], \emph{i.e.}, we treat the leader as a special player who seeks for an optimal (in terms of her utility) \emph{strategy to commit to}.
%
% Our approach is in line with previous works on OSEs computation in CGs \cite{Marchesi18:leadership,aij2018leadership}.
%
In the literature, there are a number of works that apply the Stackelberg paradigm to CGs following different approaches.
Even if these works address settings that are are substantially different from ours, it is worth discussing how their results relate to our work. 

There are some works, such as, \emph{e.g.}, \cite{roughgarden2004stackelberg,swamy2007effectiveness,sharma2009stackelberg,fotakis2010stackelberg}, which study CGs where the leader is an authority whose objective is to minimize the inefficacy (in terms of followers' social welfare) of the NE reached by the followers (\emph{i.e.}, minimize the price of anarchy). 
This setting is fundamentally different form ours, as we assume that the leader looks for a strategy to commit to that minimizes her own cost, while she is not concerned with the maximization of followers' social welfare.
Let us remark that our approach leads to what is usually called OSE, while the Stackelberg strategies analyzed in these works
% in the work by \citeauthor{roughgarden2004stackelberg}~[\citeyear{roughgarden2004stackelberg}], and its extensions, 
are not OSEs according to the classical definitions \cite{conitzer2006computing,von2010leadership}.

Moreover, there are other works, such as, \emph{e.g.}, \cite{leme2012curse,de2014sequential,correa2015curse}, which apply the Stackelberg paradigm to CGs following yet another approach.
%
%In these works, 
They assume that the players play sequentially in a predefined order, reaching a \emph{subgame perfect equilibrium} (SPE) in the extensive-form extension of the original CG where each player plays after observing the actions performed by the preceding players. 
This is different from our setting in two fundamental ways: {(i)} we assume that the followers play simultaneously, rather than sequentially; and {(ii)} these works study the inefficiency (in terms of followers' social welfare) of SPEs, rather than the computational problem of finding an optimal leader's strategy.
Furthermore, we remark that an OSE is an SPE of a particular extensive-form extension of the original CG, known as mixed extension \cite{von2010leadership}. In this extended game, the leader first commits to a mixed strategy (having a continuum of actions), and, then, the followers observe it and play simultaneously, reaching an NE. 
This is different from the extensive-form extension studied in the work by \citeauthor{leme2012curse}~[\citeyear{leme2012curse}] and its follow-ups, where only pure-strategy commitments are possible and the followers play sequentially.

%We remark that the $\mathsf{NP}$-hardness result in [CorreaEtAl15] DOES NOT IMPLY our Thm1. First, [CorreaAtAl15] studies network CGs, while our result is about CGs with arbitrary structure. Second, the problem addressed by [CorreaEtAl15] (i.e., computing a pure strategy commitment) CAN BE SOLVED IN POLY-TIME in the size of our CGs (enumerating players? pure strategies), while it is NP-H in their network CGs, as they have a super-polynomial (in graph size) number of players? pure strategies (network paths).

%% file: preliminaries.tex
\section{Preliminaries}\label{sec:preliminaries}

%In this work, 
We study \emph{Stackelberg congestion games} (SCGs) with one leader and multiple followers playing an NE in the CG resulting from the leader's commitment.
%
%These are CGs in which there is a special player who acts as a \emph{leader}, committing to a strategy beforehand, while the other players, who act as \emph{followers}, simultaneously decide how to play after observing the leader's commitment, reaching an NE in the resulting CGs.
%
%Following the notation introduced by~\citeauthor{Marchesi18:leadership}~[\citeyear{Marchesi18:leadership}],
An SCG is a tuple $(N,R,A,c_\ell,c_f)$, defined as follows.
$N=F \cup \{ \ell \}$ is a finite set of players, the leader $\ell$ and the followers $p \in F$.
$R$ is a finite set of resources.
$A= \{ A_p \}_{p \in N}$, where $A_p \subseteq 2^R$ is the set of player $p$'s actions, with each action $a_p \in A_p$ being a non-empty subset of resources, \emph{i.e}, $a_p \subseteq R$.
Finally, $c_\ell = \{c_{i,\ell}\}_{i \in R}$ and $c_f = \{c_{i,f}\}_{i \in R}$ are, respectively, the leader's and followers' cost functions, with $c_{i,\ell}, c_{i,f} : \mathbb{N} \rightarrow \mathbb{Q}$ being the costs of resource $i$ as a function of its congestion.
As usual, we assume $c_{i,\ell}(0) = c_{i,f}(0) = 0$ for every $i \in R$.
Moreover, we let $n \coloneqq |N|$ and $r \coloneqq |R|$ be the number of players and resources, respectively.

Let $\sigma_p \in \Delta_p$ be a player $p$'s strategy, which is a probability distribution over actions $A_p$, with $\sigma_p(a_p)$ denoting the probability of playing action $a_p \in A_p$. 
A strategy $\sigma_p \in \Delta_p$ is \emph{pure} when it prescribes to play a single action $a_p \in A_p$ with probability $\sigma_p(a_p) = 1$; otherwise, $\sigma_p$ is \emph{mixed}.
With an abuse of notation, given $\sigma_p \in \Delta_p$, let $\sigma_p(i) \coloneqq \sum_{ a_p \ni i} \sigma_p(a_p)$ be the probability of selecting resource $i \in R$ when $\sigma_p$ is played.
%
%A tuple of players' strategies, one per player, is called \emph{strategy profile},
%
%and, when all strategies are pure, the term \emph{action profile} usually refers to the tuple of actions played by the players.
%
%In this work, 
We denote with $\sigma = (\sigma_\ell, a)$ a \emph{strategy profile}, \emph{i.e.}, a tuple of players' strategies, in which the leader plays a (potentially) mixed strategy $\sigma_\ell \in \Delta_\ell$ and the followers play pure strategies that determine a \emph{followers' action profile} $a = (a_p)_{p \in F}$, which belongs to the set $A_F = \bigtimes_{p \in F} A_p$.

%
%the followers play pure strategies that prescribe actions in $a = (a_p)_{p \in F} \in A_F$, which is a \emph{followers' action profile}, \emph{i.e.}, a tuple of followers' actions belonging to the followers' action set $A_F = \bigtimes_{p \in F} A_p$.
%

Given a followers' action profile $a \in A_F$, we let $\nu_i^a \coloneqq | \{ p \in F \mid i \in a_p \} |$ be the number of followers selecting resource $i \in R$ in $a$.
%, i.e., the resource congestion caused by the followers' presence only.
%
For ease of presentation, we define the \emph{followers' configuration} induced by $a$ as the vector $\nu^a \in \mathbb{N}^r$ whose $i$-th component is $\nu_i^a$.
Furthermore, for $\sigma_\ell \in \Delta_\ell$, $c_{i,f}^{\sigma_\ell} : \mathbb{N} \rightarrow \mathbb{Q}$ encodes the followers' expected cost of resource $i \in R$ given $\sigma_\ell$, as a function of the number $x$ of followers selecting $i$, \emph{i.e.}, $c_{i,f}^{\sigma_\ell}(x) = \sigma_\ell(i) c_{i,f}(x+1) + (1 - \sigma_\ell(i)) c_{i,f}(x)$.
Intuitively, the followers who select resource $i$ experience a congestion incremented by one whenever the leader chooses $i$, which happens with probability $\sigma_\ell(i)$.
%
%In an SCG, 
When player $p \in N$ plays an action $a_p \in A_p$, she pays the sum of the costs of resources $i \in a_p$.
Thus, given a strategy profile $\sigma = (\sigma_\ell, a)$, $c_{p}^{\sigma} \coloneqq \sum_{i \in a_p} c_{i,f}^{\sigma_\ell}(\nu_i^a)$ is the cost of follower $p \in F$ in $\sigma$,
and $c_{\ell}^{\sigma} \coloneqq \sum_{a_\ell \in A_\ell} \sigma_\ell(a_\ell) \sum_{i \in a_\ell} c_{i,\ell}(\nu_i^a+1)$ is the leader's one.

In an SCG, after observing a leader's strategy $\sigma_\ell \in \Delta_\ell$, the followers play a new CG with resource costs $c_{i,f}^{\sigma_\ell}$, for $i \in R$.
Specifically, we assume that the followers play an NE in pure strategies~\cite{nash1951non}.
Formally, given $\sigma=(\sigma_\ell,a)$, $a$ is a (pure-strategy) NE for $\sigma_\ell$ if, for every $p \in F$ and $a_p' \in A_p$, $c_{p}^{\sigma} \leq c_{p}^{\sigma'}$, where $\sigma' = (\sigma_\ell, a')$ and $a' \in A_F$ is obtained from $a$ by replacing $a_p$ with $a_p'$.
In words, in an NE, there is no follower who has an incentive to unilaterally deviate from $a_p$ by playing another action $a_p'$. 
Given $\sigma_\ell \in \Delta_\ell$, we let $E^{\sigma_\ell}$ be the set of NEs in the followers' game for $\sigma_\ell$.
\footnote{Let us remark that we can safely restrict the attention to pure-strategy NEs, as one is guaranteed to exist in every CG~\cite{rosenthal1973class}.
Moreover, it is natural to assume that the followers will end up playing pure-strategy NEs, since they are reachable by \emph{best-response dynamics}~\cite{monderer1996potential}, where players continuously change their actions, one at a time, playing a best response to the current players' configuration.
% in the game resulting from the leader's commitment.
}

Different subclasses of SCGs can be defined by making additional assumptions on their elements.
For instance, one possibility is to require that players' cost functions be \emph{(weakly) monotonic} in the resource congestion, \emph{i.e.}, for every resource $i \in R$, it must be the case that $c_{i,\ell}(x) \leq c_{i,\ell}(x+1)$ and $c_{i,f}(x) \leq c_{i,f}(x+1)$ for all $x \in \mathbb{N}$.

Another possibility is to restrict the structure of the players' action sets $A_p$.
Along this direction, we address \emph{Stackelberg singleton congestion games} (SSCGs), which are SCGs with players' actions required to be singletons, \emph{i.e.}, $|a_p|=1$ for every $p \in N$ and $a_p \in A_p$.
Thus, when studying such games, we identify actions with resources, assuming $A_p \subseteq R$ for all $p \in N$.
We also use $\sigma_\ell \in \Delta_\ell$ as if it were directly defined over resources in $A_\ell$, with $\sigma_\ell(i)$ being the probability of playing resource $i \in A_\ell$.
Moreover, given $\sigma = (\sigma_\ell, a)$, we have $c_{p}^{\sigma} = c_{a_p,f}^{\sigma_\ell}(\nu_{a_p}^a)$ and $c_{\ell}^{\sigma} = \sum_{i \in A_\ell} \sigma_\ell(i) c_{i,\ell}(\nu_i^a+1)$.
%
% Moreover, given $\sigma = (\sigma_\ell, a)$, we have $c_p^\sigma = c_{a_p, f}^{\sigma_\ell} (\nu_{a_p}^a)$ and $c_{\ell}^{\sigma} = \sum_{i \in A_\ell} \sigma_\ell(i) c_{i,\ell}(\nu_i^a+1)$.

Finally, we can consider different kinds of players' structures.
Here, we focus on two notable cases.
In the first one, all players share the same set of actions, \emph{i.e.}, $A_p = A \subseteq 2^R$ for all $p \in N$.
We refer to these games as \emph{symmetric}.
Instead, in the second case, there exists a finite set $\mathcal{T} = \{ 1, \ldots, T \}$ of \emph{followers' classes}, with followers of the same class sharing the same set of actions. 
We say that these games are \emph{$\mathcal{T}$-class}.
Specifically, in a $\mathcal{T}$-class SCG, we can partition the followers into $T$ disjoint sets, \emph{i.e.}, $F = \bigcup_{t \in \mathcal{T}} F_t$, so that, for each $t \in \mathcal{T}$, $A_p = A_t \subseteq 2^R$ for all $p \in F_t$.
We also let $|F_t|=n_t$ be the number of followers of class $t \in \mathcal{T}$.
When studying these games, given a followers' action profile $a \in A_F$, we let $\nu_{t,i}^a \coloneqq |\{ p \in F_t \mid i \in a_p \}|$ be the number of followers of class $t \in \mathcal{T}$ selecting resource $i \in R$ in $a$.
Moreover, we define the \emph{followers' configuration of class $t$} induced by $a$ as the vector $\nu_t^a \in \mathbb{N}^r$ whose $i$-th component is $\nu_{t,i}^a$.
\footnote{Let us remark that symmetric SCGs are a special case of $\mathcal{T}$-class SCGs with only one class, \emph{i.e.}, $\mathcal{T}=\{ 1 \}$, and leader's action set equal to the followers', \emph{i.e.}, $A_\ell = A_1$.}

Observe that $\mathcal{T}$-class SSCGs
%, where players' actions are singletons and followers are  partitioned into $K$ types, 
can be fully analyzed using followers' configurations, rather than action profiles.
This is because only the number of followers of each class selecting each resource is significant, and,
thus, a followers' action profile $a \in A_F$ can be equivalently represented with the followers' configurations $\{ \nu_t^a \}_{t \in \mathcal{T}}$.
%$\nu_t^a$ induced by it, one per type $t \in \mathcal{T}$.
%
Thus, we can directly use the vector $\nu_t \in \mathbb{N}^r$ with $\sum_{i \in R} \nu_{t,i} = n_t$ to denote a followers' configuration of class $t\in \mathcal{T}$.
Moreover, given $\{ \nu_t \}_{t \in \mathcal{T}}$, we let $\nu \in \mathbb{N}^r$ be such that $\nu_i \coloneqq \sum_{t \in \mathcal{T}} \nu_{t,i}$ for $i \in R$.
Notice that, given a strategy profile $\sigma =(\sigma_\ell, \{ \nu_t \}_{t \in \mathcal{T}})$, $\{ \nu_t \}_{t \in \mathcal{T}}$ is an NE for $\sigma_\ell $ if, for every class $t \in \mathcal{T}$, and resources $i \in A_t : \nu_{t,i} > 0$ and $j \in A_t$, it holds $c_{i,f}^{\sigma_\ell}(\nu_i) \leq c_{j,f}^{\sigma_\ell}(\nu_j + 1)$.
%
%
%Thus, all the ab discussion about $\mathcal{K}$-type SSCGs holds for symmetric SSCGs as well.
%

In conclusion, in an SCG,
%as most of the works in the literature (see, for instance, \cite{conitzer2006computing}), 
% we study \emph{optimistic SEs} (OSEs), which are SEs where the leader assumes that the followers always break ties in her favor
%
an OSE prescribes the leader an optimal strategy to commit to, assuming that the followers play a pure-strategy NE minimizing her cost in the CG resulting from the commitment.
Formally, $\sigma = (\sigma_\ell, a)$ is an OSE if it solves the problem:
% is an optimal solution to the following problem: 
$\min_{\sigma_\ell' \in \Delta_\ell} \min_{a' \in E^{\sigma_\ell'}} c_\ell^{\sigma' = (\sigma_\ell',a')}$.
%
%minimizes the leader's cost $c_\ell^{\sigma'}$ over all strategy profiles $\sigma' = (\sigma_\ell', a')$ such that $a' \in E^{\sigma_\ell'}$.

%% file: congestion_games.tex
\section{Complexity Results on SCG}\label{sec:congestion_games}

First, we study the problem of computing an OSE in SCGs, and we prove it is intractable even when the game is symmetric, cost functions are monotonic, and players' actions are made of only two resources.
Thus, non-singleton actions make the problem considerably harder than in the singleton case, which admits a polynomial-time algorithm in symmetric games with monotonic costs~\cite{Marchesi18:leadership}.

In particular, we show that the problem is \textsf{NP}-hard and not in Poly-\textsf{APX} unless \textsf{P} = \textsf{NP}, \emph{i.e.}, the leader's cost in an OSE cannot be efficiently approximated up to any polynomial factor in the size of the input.
Our results are based on a reduction from \textsc{3SAT}, a well-know \textsf{NP}-complete problem~\cite{garey1979computers}.
%, which reads as follows.
%
Intuitively, given $\epsilon > 0$, we map any \textsc{3SAT} instance to an SCG that admits an OSE $\sigma$ with $c_\ell^\sigma = \epsilon$ if and only if \textsc{3SAT} is satisfiable, otherwise $c_\ell^\sigma = 1$.

%\begin{definition}[3SAT]\label{def:3sat}
%	Given a finite set $C$ of 3-literal clauses defined over a finite set $U$ of variables, is there a truth assignment to the variables which satisfies all clauses? 
%	%
%	\footnote{ We refer to a 3SAT instance as $(C,U)$.
%		%
%		Moreover, $l \in \phi$ denotes a literal (\emph{i.e.}, a variable or its negation) appearing in $\phi \in C$, while $u(l)$ is the variable corresponding to that literal.
%		%
%		%We also let $|C|=m$ and $|V|=s$ be, respectively, the number of clauses and variables.
%		}
%\end{definition}

\begin{restatable}{theorem}{theoremone}\label{thm:hard_general_symm}
	Computing an OSE in symmetric SCGs is \textsf{NP}-hard, even when cost functions are monotonic and players' actions have cardinality at most two.
	%with the same action space for each player, monotonic cost function and actions with cardinality $2$, and leader playing in pure strategies is NP-Hard.
\end{restatable}

By letting $\epsilon = \frac{1}{2^I}$, where $I$ is the game size, the reduction used for Theorem~\ref{thm:hard_general_symm} also shows the following:

\begin{restatable}{theorem}{theoremtwo}\label{thm:inapx_general_symm}
	The problem of computing an OSE in symmetric SCGs is not in Poly-\textsf{APX} unless \textsf{P} = \textsf{NP}, even when costs are monotonic and players' actions have cardinality two.
\end{restatable}

In conclusion, we provide some side results that deepen our analysis on how non-singleton actions impact on the complexity of finding an OSE in SCGs.
The following theorem shows that our intractability results hold even in SCGs where only the followers have non-singleton actions.
%
%The following theorem shows that the intractability results in Theorems~\ref{thm:hard_general_symm}~and~\ref{thm:inapx_general_symm} hold even in SCGs with monotonic costs where only the followers have non-singleton actions.
%

\begin{restatable}{theorem}{theoremthree}\label{thm:inapx_general_follower_non_sing}
	The problem of computing an OSE in SCGs is \textsf{NP}-hard and not in Poly-\textsf{APX} unless \textsf{P} = \textsf{NP}, even when leader's actions are singletons, costs are monotonic, and followers are symmetric with actions of cardinality at most two.
\end{restatable} 

%%%%\vspace{-3mm}
%%%%\begin{proof}
%%%%	The result is readily obtained from the proofs of Theorems~\ref{thm:hard_general_symm}~and~\ref{thm:inapx_general_symm}, by setting $A_\ell = \{a_w\}$ in the reduction.
%%%%	%
%%%%	%In order to prove the result, 
%%%%	%It is sufficient to change the reduction in the proof of Theorem~\ref{thm:hard_general_symm}, setting $A_\ell = \{a_w\}$, while all the rest remains the same.
%%%%	%
%%%%	%Since, in these new games, the leader has to play $a_w$, the proof is like those of Theorems~\ref{thm:hard_general_symm}~and~\ref{thm:inapx_general_symm}.
%%%%	%
%%%%\end{proof}

Since the OSEs of the games used in our reduction prescribe the leader to play a pure strategy,
we have that:
%, the same intractability results also hold in the special case in which the leader is restricted to pure strategies.

\begin{observation}
	The results in Theorems~\ref{thm:hard_general_symm},~\ref{thm:inapx_general_symm},~and~\ref{thm:inapx_general_follower_non_sing} hold even if we restrict the leader to play pure strategies.
\end{observation}

%Let us remark that, since in the games of the reduction the leader has a single action available, Theorem~\ref{thm:inapx_general_follower_non_sing} holds even if we restrict the leader to play pure strategies.

Finally, let us consider the case in which only the leader has non-singleton actions.
We have the following observation.

\begin{observation}
	In SCGs with symmetric followers having singleton actions, an OSE can be found in polynomial time if we restrict the leader to play pure strategies.~\footnote{ We leave as an open problem the study of the computational complexity when the leader is allowed to play mixed strategies.}
\end{observation}

A polynomial-time algorithm enumerates the leader's pure strategies, and, for each of them, it computes an NE minimizing the leader's cost in the resulting followers' symmetric singleton CG, which can be done in polynomial time using dynamic programming, as shown by~\citeauthor{ieong2005fast}~[\citeyear{ieong2005fast}].

%% file: singleton_games.tex
\section{Complexity Results on $\mathcal{T}$-class  SSCG}\label{sec:singleton_games}

%In this section, 
We switch the attention to SSCGs, \emph{i.e.}, games where players' actions are singletons.
In particular, we focus on $\mathcal{T}$-class SSCGs, a generalization of symmetric SSCGs.
%; the latter already being thoroughly studied in the literature~\cite{Marchesi18:leadership,aij2018leadership}.
%
As shown by~\citeauthor{aij2018leadership}~[\citeyear{aij2018leadership}], finding an OSE in symmetric SSCGs with non-monotonic costs is \textsf{NP}-hard, while the problem becomes easy if: {(a)} we assume that the leader can only play pure strategies~\cite{aij2018leadership}, or {(b)} we force players' costs be monotonic~\cite{Marchesi18:leadership}.

Here, first we show that, under condition {(a)}, computing an OSE is easy also in $\mathcal{T}$-class SSCGs with a fixed number of classes.
Next, we prove that, if condition {(a)} does not hold, then the problem is \textsf{NP}-hard in $\mathcal{T}$-class SSCGs even if we enforce {(b)} and there are only four followers' classes.

Let us start providing a polynomial-time algorithm for computing OSEs in $\mathcal{T}$-class SSCGs with a fixed number of classes and the leader restricted to pure strategies.
%
%As usual, 
Our algorithm iterates over leader's pure strategies, \emph{i.e.}, resources $i \in A_\ell$, and, for each of them, it finds an NE minimizing the leader's cost in the followers' game for $\sigma_\ell \in \Delta_\ell : \sigma_\ell(i) = 1$.
Such NE is computed with an extension of the \emph{dynamic programming} algorithm proposed by~\citeauthor{ieong2005fast}~[\citeyear{ieong2005fast}] for finding an optimal NE in symmetric singleton CGs without leadership.
We extend the algorithm to $\mathcal{T}$-class singleton CGs without leadership, which are defined as their Stackelberg counterpart, except for the fact that there is no leader and all players experience the same resource costs $c = \{ c_i \}_{i \in R}$, with $c_{i} : \mathbb{N} \rightarrow \mathbb{Q}$ defined as usual.

%Following the presentation of~\citeauthor{aij2018leadership}~[\citeyear{aij2018leadership}], 
We define $\mathcal{O}(h,B_1,\ldots,B_T,M_1,\ldots,M_T,V_1,\ldots,V_T)$ as the cost of an optimal (according to some optimality criterion) NE for a $\mathcal{T}$-class singleton CG restricted to $i$ resources $\{1,2,\ldots,i\} \subseteq R$ and $B_t$ players for each class $t \in \mathcal{T}$, where $M_t$ is the largest cost experienced by a player of class $t$ and $V_t$ is the smallest cost a player of class $t$ would get by switching to another resource.
When the optimality criterion is the minimization of the sum of players' costs, we have:

\begin{restatable}{lemma}{lemmaone}\label{lem:dp}
	$\mathcal{O}(i,B_1,\ldots,B_T,M_1,\ldots,M_T,V_1,\ldots,V_T)$ satisfies the following recursive equation:
	%
	%\vspace{-1mm}
	\begin{small}
		\begin{align}
		& \mathcal{O}(i,B_1,\dots,B_T,M_1,\dots,M_T,V_1,\dots,V_T) = \nonumber \\ 
		& \hspace{-2mm} \min_{\substack{ \phantom {A} \\ p_t \in \{0, \ldots, B_t\} \ \forall t \in \mathcal{T} \\ m_t \in \{0, \ldots, M_t\} \ \forall t \in \mathcal{T} \\ v_t \in \{1, \ldots, V_t\} \ \forall t \in \mathcal{T}}} \hspace{-8mm}
		\mathcal{O}(i-1,p_1,\ldots,p_T,m_1,\ldots,m_T,v_1,\ldots,v_T) + b c_i(b) \nonumber 
		\end{align}
		\vspace{-2mm}
		\begin{subequations}
			\begin{align}
			\textnormal{s.t.} \quad& \label{dp:12}   b=\sum_{t \in \mathcal{T}} (B_t -p_t) \\
			\label{dp:15}
			&B_t=p_t & \forall t \in \mathcal{T} : i \notin A_t \\
			\label{dp:14}
			& m_t \leq M_t & \forall t \in \mathcal{T} \\
			\label{dp:13}
			& v_t \geq V_t & \forall t \in \mathcal{T} \\
			\label{dp:2}
			& c_i(b) \leq M_t & \forall t \in \mathcal{T} : B_t-p_t > 0 \\
			\label{dp:4}
			& c_i(b+1) \geq V_t & \forall t \in \mathcal{T} : i \in A_t \\
			\label{dp:3}
			& c_i(b) \leq v_t & \forall t \in \mathcal{T} : B_t-p_t > 0 \\
			\label{dp:5}
			& c_i(b+1) \geq m_t & \forall t \in \mathcal{T} : i \in A_t.
			\end{align}
		\end{subequations}
	\end{small}
\end{restatable}

\begin{remark}
	Lemma~\ref{lem:dp} holds also for other optimality criteria, e.g., the minimization of the cost of a given resource. 
\end{remark}

Thus, we can conclude the following:

\begin{restatable}{theorem}{theoremfour}\label{thm:dyn}
	In $\mathcal{T}$-class singleton CGs without leadership, an optimal (given an optimality criterion) NE can be found in $O(n^{6T} r^{4T+1})$. 
	In $\mathcal{T}$-class SSCGs, an OSE can be found in $O(n^{6T} r^{4T+2})$ if we restrict the leader to play pure strategies.
\end{restatable}
%
%\vspace{-3mm}
%%%%\begin{proof}
%%%%	Since there are at most $n r$ different values of costs $c_i(x)$ ($i \in R$, $x \in \{1,\ldots,n\}$), for each $t \in \mathcal{T}$ there are at most $nr$ values of $M_t$ and $V_t$.
%%%%	%
%%%%	There are also exactly $r$ values of $i \in R$ and exactly $n_t$ values of $B_t$ for each $t \in \mathcal{T}$.
%%%%	%, with $\sum_{t \in \mathcal{T}} n_t = n-1$.
%%%%	%
%%%%	Hence, the dynamic programming table of $\mathcal{O}(i,B_1,\dots,B_T,M_1,\dots,M_T,V_1,\dots,V_T)$ contains $O(n^{3T} r^{2T+1})$ entries.
%%%%	%
%%%%	Computing an entry of the table requires $O(n^{3T} r^{2T})$.
%%%%	%
%%%%	Overall, an optimal NE is computed in $O(n^{6T} r^{4T+1})$.
%%%%	%
%%%%	For $\mathcal{T}$-type SSCGs with the leader restricted to play pure strategies, it suffices to run the algorithm for each $i \in A_\ell$, using as optimality criterion the minimization of the cost of $i$.
%%%%	%
%%%%	Since there are $O(r)$ leader's actions, the overall complexity is $O(n^{6T} r^{4T+2})$.
%%%%	%
%%%%\end{proof}

%\vspace{-3mm}
\begin{corollary}
	In $\mathcal{T}$-class SSCGs with a fixed number of classes, an OSE can be found in polynomial time if we restrict the leader to play pure strategies.
\end{corollary}

Now, we prove the hardness result, using a reduction from $K$-\textsc{PARTITION}, an \textsf{NP}-compete variant of \textsc{PARTITION} with an additional size constraint~\cite{aij2018leadership}.
	%
%	\footnote{You can find the complete proof of Theorem~\ref{thm:np_hard_types} in the Supplemental Material, here we only report its main ideas.}

%\begin{definition}[$K$-PARTITION]
%	Given a finite set $S = \{ s_1, \ldots, s_{|S|} \}$ of positive integers $s_i \in \mathbb{N}^+$ with $\sum_{s_i \in S} s_i = 2X$ and a positive integer $K \leq \frac{|S|}{2}$, is there a subset $S' \subseteq S$ such that $|S'|=K$ and $\sum_{s_i \in S'} s_i = X$?
%\end{definition}

\begin{restatable}{theorem}{thmpart}\label{thm:np_hard_types}
	Computing an OSE in $\mathcal{T}$-class SSCGs is \textsf{NP}-hard, even when cost functions are monotonic and $|\mathcal{T}|=4$.
\end{restatable}

%% file: formulations.tex
\section{Mathematical Programming Formulations}\label{sec:formulations}

In this section, we provide two MILP formulations for finding OSEs in SCGs.
The first one is specifically tailored to $\mathcal{T}$-class SSCGs, while the second one works for general SCGs.
%
%Our formulations generalize those provided by~\citeauthor{aij2018leadership}~\shortcite{aij2018leadership} for SSCGs and their symmetric variant.

%\subsection{Formulation for $\mathcal{T}$-type SSCGs} 

Let us start with $\mathcal{T}$-class SSCGs.
%
 %For ease of presentation, for $i \in R$ and $t \in \mathcal{T}$, let $V(t,i) = \{0,1, \dots, n_t \}$ be the set of possible congestion levels of resource $i$ of followers of type $t$ on a resource $i$ and $V(i)= \sum_{t \in T} v(t,i)$ the possible congestion levels on a resource $i$.
 %
 For ease of presentation, let $V(i) \coloneqq \{ 1, \ldots, v_i^{\max} \}$ be the set of possible congestion levels induced by the followers on resource $i \in R$, with $v_i^{\max} \coloneqq \sum_{t \in \mathcal{T}: i \in A_t} n_t$.
 Moreover, let $V(t) \coloneqq \{ 1, \ldots, n_t \}$ for $t \in \mathcal{T}$.
 %
 % Let, for each type of player $t \in T=\{1,\dots,K\}$, $n_t$ be the number of players of type $t$ and $A_t$ be the resources available to followers of type $t$.
 %
 For every class $t \in \mathcal{T}$, resource $i \in A_t$, and value $v \in V(t)$, let us introduce the binary variable $q_{t i v}$, which is equal to 1 if and only if $v$ followers of class $t$ select resource $i$. 
 Furthermore, for every $i \in R$ and $v \in V(i)$, let the binary variable $y_{i v}$ be equal to 1 if and only if $v$ followers select resource $i$. 
 These variables represent followers' configurations.
 Specifically, for $t \in \mathcal{T}$, $\nu_t \in \mathbb{N}^r$ is such that $\nu_{t,i} = \sum_{v \in V(t)} v \, q_{t i v}$ for all $i \in R$, while $\nu \in \mathbb{N}^r$ is such that $\nu_{i} = \sum_{v \in V(i)} v \, y_{i v}$ for all $i \in R$.
 %
 %$\nu \in \mathbb{N}^r$, namely, $\nu_i = \sum_{v \in V} v \, y_{i v}$ for all $i \in R$.
 %
 For every $i \in R$, let $\alpha_i \in [0,1]$ be equal to $\sigma_\ell(i)$,
 and, for $v \in V(i)$, let the auxiliary variable $z_{i v}$ be equal to the bilinear term $y_{i v} \alpha_i$.
 Finally, let $M > \max\{c_{i,f}(v) \mid i \in R, v \in \{ 1, \ldots, v_i^{\max} +2 \} \}$.
 
 The complete MILP formulation reads as follows:

 \begin{scriptsize}
 \begin{subequations}\label{eq:milp_dif_types}
 	\begin{align}
 	\min & \quad \sum_{i \in R} \ \sum_{v \in V(i) } c_{i, \ell}(v+1) \,  z_{i v} \label{eq:milp_dif_types_of} \\
 	\textnormal{s.t.} & \sum_{v \in V(t) } q_{t i v} \leq 1 \hspace{3.1cm} \forall t \in \mathcal{T}, \forall i \in A_t \label{eq:milp_dif_types_cons1} \\
	%&q_{t i v}=0  \forall t \in T,\forall i \in R, \forall v \in V | i \notin A_t \label{eq:milp_dif_types_cons11}\\
 	& \sum_{v \in V(i) } y_{i v} \leq 1 \hspace{4.25cm}  \forall i \in R  \label{eq:milp_dif_types_cons10} \\
 	& \sum_{i \in A_t} \ \sum_{v \in V(t)} v \, q_{t i v}  = n_t \hspace{3.1cm}   \forall t \in \mathcal{T}  \label{eq:milp_dif_types_cons2}\\
 	& \sum_{t \in \mathcal{T} : i \in A_t} \sum_{v \in V(t)} v \, q_{t i v} = \sum_{v \in V(i)}v \, y_{i v} \hspace{1.5cm}  \forall i \in R \label{eq:milp_dif_types_cons9} \\
 	& \sum_{v \in V(j)} \left(  y_{j v} c_{j,f}(v+1) + z_{j v} \Big( c_{j,f}(v+2) - c_{j,f}(v+1)  \Big)  \right) \geq \nonumber \hspace{-2cm}\\
 	& \sum_{v \in V(i)} \left(  y_{i v} c_{i,f}(v) + z_{i v} \Big( c_{i,f}(v+1) - c_{i,f}(v)  \Big)  \right) + \nonumber \\
 	& \hspace{1cm}- M \Big(1- \sum_{v \in V(t)} q_{t i v} \Big) \hspace{0.7cm} \forall t \in \mathcal{T}, \forall i \neq j \in A_t \label{eq:milp_dif_types_cons3} \\
 	& z_{i v} \leq \alpha_i \hspace{3.7cm}  \forall i \in R, \forall v \in V(i) \label{eq:milp_dif_types_cons4} \\
 	& z_{i v} \leq y_{i v} \hspace{3.60cm} \forall i \in R, \forall v \in V(i) \label{eq:milp_same_actions_cons5} \\
 	& z_{i v} \geq \alpha_i + y_{i v} -1 \hspace{2.47cm} \forall i \in R, \forall v \in V(i) \label{eq:milp_same_actions_cons6} \\
 	& z_{i v} \geq 0 \hspace{3.85cm} \forall i \in R, \forall v \in V(i) \label{eq:milp_same_actions_cons8}\\
 	& \sum_{i \in R} \alpha_i = 1 \label{eq:milp_same_actions_cons7}\\
 	& \alpha_i \geq 0 \hspace{5.2cm}  \forall i \in R \\
 	& \alpha_i = 0 \hspace{5.1cm}  \forall i \notin A_\ell \\
 	& q_{t i v} \in \{ 0,1 \} \hspace{2.4cm}  \forall t \in \mathcal{T}, i \in A_t,  v \in V(t) \\
 	& y_{i v} \in \{ 0,1 \} \hspace{3.2cm}  \forall i \in R, \forall v \in V(i).
 	\end{align}
 \end{subequations}
 \end{scriptsize}

% \vspace{-3mm}
 %
 Function~\eqref{eq:milp_dif_types_of} is the leader's expected cost.
 Constraints~\eqref{eq:milp_dif_types_cons1} ensure that at most one variable $q_{t i v}$ be equal to $1$ for each class $t \in \mathcal{T}$ and resource $i \in A_t$, and, thus, the number of followers of class $t$ on each resource is uniquely determined (note that $\sum_{v \in V(t)} q_{t i v} = 0$ if no follower of class $t$ selects resource $i$).
  Constraints~\eqref{eq:milp_dif_types_cons10} ensure that at most one variable $y_{i v}$ be equal to $1$ for each resource $i \in R$, which guarantees that the congestion level of each resource is uniquely determined.
 Constraints~\eqref{eq:milp_dif_types_cons2} guarantee that the followers' configuration be well-defined, \emph{i.e.}, for all $t \in \mathcal{T}$, exactly $n_t$ followers of class $t$ are present.
 Constrains~\eqref{eq:milp_dif_types_cons9} ensure that the congestion level on resource $i \in R$ be equal to the sum of the congestion levels induced by all classes. 
 Constraints~\eqref{eq:milp_dif_types_cons3} force the followers' configurations defined by the $q_{t i v}$ variables be an NE for the leader's strategy identified by the $\alpha_i$ variables.
 This follows from the fact that, being $z_{iv} = y_{i v} \alpha_i $ and $z_{j v} =  y_{j v} \alpha_j$, the right-hand side is the cost incurred by the followers of class $t \in \mathcal{T}$ selecting resource $i \in A_t$, while the left-hand side is the cost they would incur after deviating to $j \neq i \in A_t$.
 %
 % This follows from the fact that $\displaystyle \sum_{v \in V(i)} \left(  y_{i v} c_{i,f}(v) + z_{i v} \Big( c_{i,f}(v+1) - c_{i,f}(v)  \Big)  \right)$ (recall that $z_{iv} = y_{i v} \alpha_i $) is equal to the cost incurred by the followers who select resource $i \in R$, while $\displaystyle \sum_{v \in V} \left(  y_{j v} c_{j,f}(v+1) + z_{j v} \Big( c_{j,f}(v+2) - c_{j,f}(v+1)  \Big)  \right)$ (recall that $z_{j v} =  y_{j v} \alpha_j$) is equal to the cost they would incur after deviating to resource $j \in R$.
 %
 Note that, for each $t \in \mathcal{T}$, the constrains are active only if there is at least one follower of class $t$ selecting $i$. 
 %
 %Let us remark that Constraints~\eqref{milp_dif_types_cons3} are trivially satisfied if $y_{i v}= 0$ for all $v \in V$. This is correct as, if no followers choose resource $i \in R$, no equilibrium conditions need to be enforced.
 %
 Finally, Constraints~\eqref{eq:milp_dif_types_cons4}--\eqref{eq:milp_same_actions_cons8} are McCormick envelope constraints~\cite{mccormick1976computability} which guarantee $z_{i v} = y_{i v} \alpha_i$ when $y_{iv} \in \{0,1\}$.
 
% We remark that Formulation~\eqref{eq:milp_dif_types} features $r (2n + 1)$ variables, $n r$ of which binary, and $r (r - 1) + r (3n + 1) + 2$ constraints.

%\subsection{Computing an OSE in SCGs}\label{sub_sec:milp_same_actions}

Next, we extend Formulation~\eqref{eq:milp_dif_types} to deal with general SCGs.
Letting $ v_i^{\max} \coloneqq | \{ p \in F \mid \exists a_p \in A_p : i \in a_p \}  |$ be the maximum number of followers who can select resource $i \in R$, we define $V(i) \coloneqq \{1, \ldots, v_i^{\max} \}$.
For every follower $p \in F$ and action $a_p \in A_p$, we introduce the binary variable $x_{p \, a_p}$, which is equal to 1 if and only $p$ plays $a_p$.
Moreover, for every $a_\ell \in A_\ell$, let $\alpha_{a_\ell} \in [0,1]$ be equal to $\sigma_\ell(a_\ell)$.
All the variables already defined in~Formulation~\eqref{eq:milp_dif_types} are used here with the same meaning.
Finally, let $M > \sum_{i \in R} \max\{c_{i,f}(v) \mid i \in R, v \in \{ 1, \ldots, v_i^{\max} +2 \} \}$.

The complete MILP formulation reads as follows:

%\vspace{-3mm}
\begin{scriptsize}
\begin{subequations}\label{eq:milp_general}
	\begin{align}
		\min & \quad \sum_{i \in R} \ \sum_{v \in V(i) } c_{i, \ell}(v+1) \,  z_{i v} \label{eq:milp_non_singleton_of} \\
   \textnormal{s.t.} & \sum_{a_p \in A_p} x_{p \, a_p} = 1 \hspace{4cm} \forall p \in F \label{eq:milp_non_singleton_cons0} \\
                & \sum_{v \in V(i) } y_{iv} \leq 1 \hspace{4.25cm} \forall i \in R \label{eq:milp_non_singleton_cons1} \\
		&  \sum_{v \in V(i)} v \, y_{i v}  = \sum_{p \in F} \ \sum_{a_p \in A_p : i \in a_p} x_{p \, a_p} \hspace{1.5cm} \forall i \in R  \label{eq:milp_non_singleton_cons2}\\
		& \hspace{-6mm} \sum_{i \in a_p' \setminus a_p} \sum_{v \in V(i)} \left(  y_{i v} c_{i,f}(v+1) + z_{i v} \Big( c_{i,f}(v+2) - c_{i,f}(v+1)  \Big)  \right) \geq \nonumber \hspace{-5cm}\\
   & \hspace{-6mm} \sum_{i \in a_p \setminus a_p'} \sum_{v \in V(i)} \left(  y_{i v} c_{i,f}(v) + z_{i v} \Big( c_{i,f}(v+1) - c_{i,f}(v)  \Big)  \right) + \nonumber \\
   & \hspace{1.3cm}- M \Big(1- x_{p \, a_p} \Big) \hspace{0.85cm} \forall p \in F, a_p \neq a_p' \in A_p \label{eq:milp_non_singleton_cons3} \\
		& z_{i v} \leq \sum_{a_\ell \in A_\ell : i \in a_\ell} \alpha_{a_\ell} \hspace{2.15cm} \forall i \in R, \forall v \in V(i) \label{eq:milp_different_actions_cons4} \\
		& z_{i v} \leq y_{i v} \hspace{3.6cm} \forall i \in R, \forall v \in V(i) \label{eq:milp_different_actions_cons5} \\
		& z_{i v} \geq \sum_{a_\ell \in A_\ell : i \in a_\ell}\alpha_{a_\ell} + y_{i v} -1 \hspace{0.95cm} \forall i \in R, \forall v \in V(i) \label{eq:milp_different_actions_cons6} \\
		& z_{i v} \geq 0 \hspace{3.85cm} \forall i \in R, \forall v \in V(i) \label{eq:milp_different_actions_cons8}\\
		& \sum_{a_\ell \in A_\ell} \alpha_{a_\ell} = 1 \label{eq:milp_different_actions_cons7}\\
		& \alpha_{a_\ell} \geq 0 \hspace{4.8cm} \forall {a_\ell} \in A_\ell \label{eq:milp_non_singleton_cons77}\\
        & x_{p \, a_p} \in \{ 0,1 \} \hspace{3.1cm} \forall p \in F, \forall a_p \in A_p \\
		& y_{i v} \in \{ 0,1 \} \hspace{3.2cm} \forall i \in R, \forall v \in V(i).
	\end{align}
\end{subequations}
\end{scriptsize}

%\vspace{-3mm}
%
Function~\eqref{eq:milp_non_singleton_of}, Constraints~\eqref{eq:milp_non_singleton_cons1}, and Constraints~\eqref{eq:milp_non_singleton_cons3}--\eqref{eq:milp_non_singleton_cons77} have the same meaning as their counterparts in Formulation~\eqref{eq:milp_dif_types}.
Note that, in this case, $z_{i v} = y_{i v} \,  \sum_{a_\ell \in A_\ell : i \in a_\ell}\alpha_{a_\ell} $, where the summation represents the probability $\sigma_\ell(i)$ and $\sigma_\ell$ is identified by the $\alpha_{a_\ell}$ variables.
Constraints~\eqref{eq:milp_non_singleton_cons0} ensure that each follower selects exactly one action.
Constraints~\eqref{eq:milp_non_singleton_cons2} guarantee that the followers' configuration represented by the variables $y_{i v}$ be well-defined.
%, i.e., for every $i \in R$, $\nu_i = \sum_{v \in V(i)} v \, y_{iv}$ is equal to $\sum_{p \in F} \sum_{a_p \in A_p : i \in a_p} x_{p \, a_p}$, which is the number of followers selecting resource $i$.
%
Note that, 
%differently from Formulation~\eqref{eq:milp_dif_types}, 
Constraints~\eqref{eq:milp_non_singleton_cons3} are enforced for each follower $p \in F$ here, and they are active only for the action $a_p \in A_p$ that she plays.
Constraints~\eqref{eq:milp_non_singleton_cons3} do not account for the costs of resources $i \in a_p \cup a_p'$, since they do not change when deviating from $a_p$ to $a_p'$.

%% file: experiments.tex
 \begin{figure*}[!htp]
	%\vspace{1mm}
%	\begin{minipage}{4.3cm}
%		\includegraphics[scale=.44]{figures/figure1_types.pdf}
%		\subcaption{$\mathcal{T}$-class SSCGs $(n_t = 0.2\,r)$}
%		\begin{center}
%			(d) SCGs $(n = r)$
%		\end{center}
%	\end{minipage}
	\begin{minipage}{4.3cm}
		\includegraphics[scale=.38]{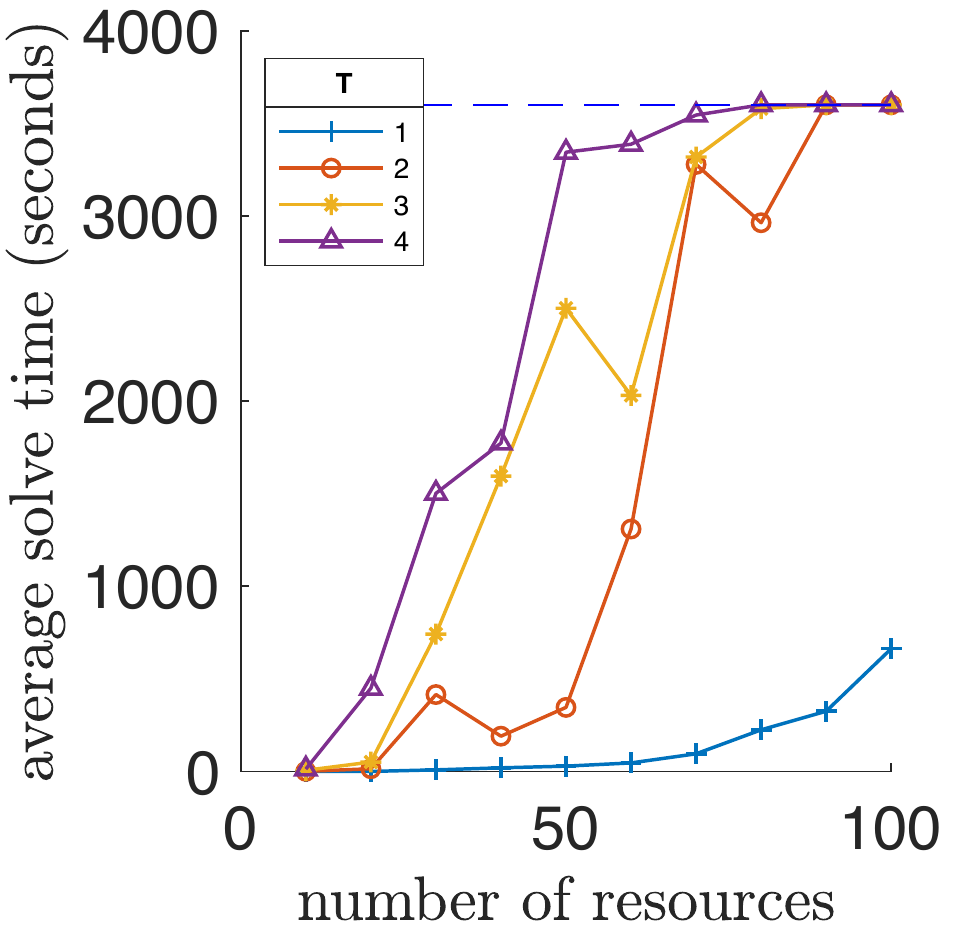}
		\subcaption{$\mathcal{T}$-class SSCGs $(n_t = 0.5\,r)$}
%		\begin{center}
%			(e) SCGs $(n = 2\,r)$
%		\end{center}
	\end{minipage}
%	\begin{minipage}{4.3cm}
%		\includegraphics[scale=.44]{figures/figure3_types.pdf}
%		\subcaption{$\mathcal{T}$-class SSCGs $(n_t = r)$}
%		\begin{center}
%			(f) SCGs $(n = 3\,r)$
%		\end{center}
%	\end{minipage}
	\begin{minipage}{4.3cm}
		\includegraphics[scale=.38]{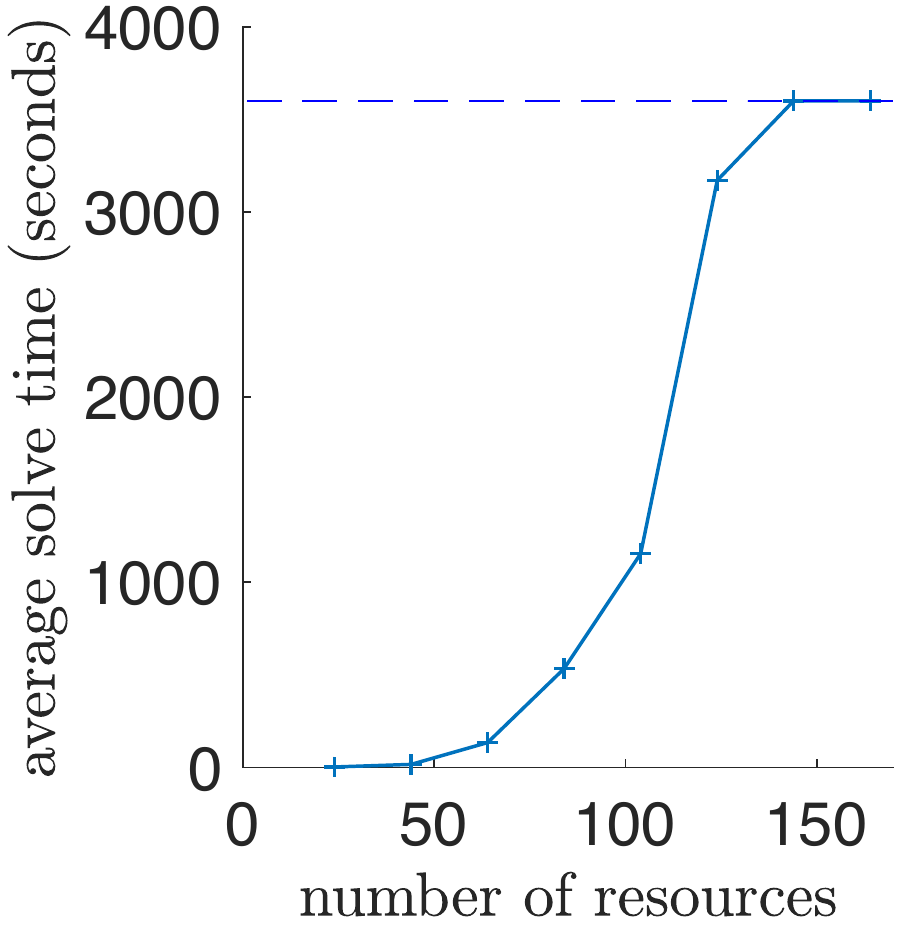}
		\subcaption{Worst-case $\mathcal{T}$-class SSCGs}
%		\begin{center}
%			(f) SCGs $(n = 3\,r)$
%		\end{center}
	\end{minipage}
%	\begin{minipage}{4.3cm}
%		\includegraphics[scale=.44]{figures/figure1_non_sin}
%		\subcaption{ SCGs $(n = r)$}
%		\begin{center}
%			(d) SCGs $(n = r)$
%		\end{center}
%	\end{minipage}
	\begin{minipage}{4.3cm}
		\includegraphics[scale=.38]{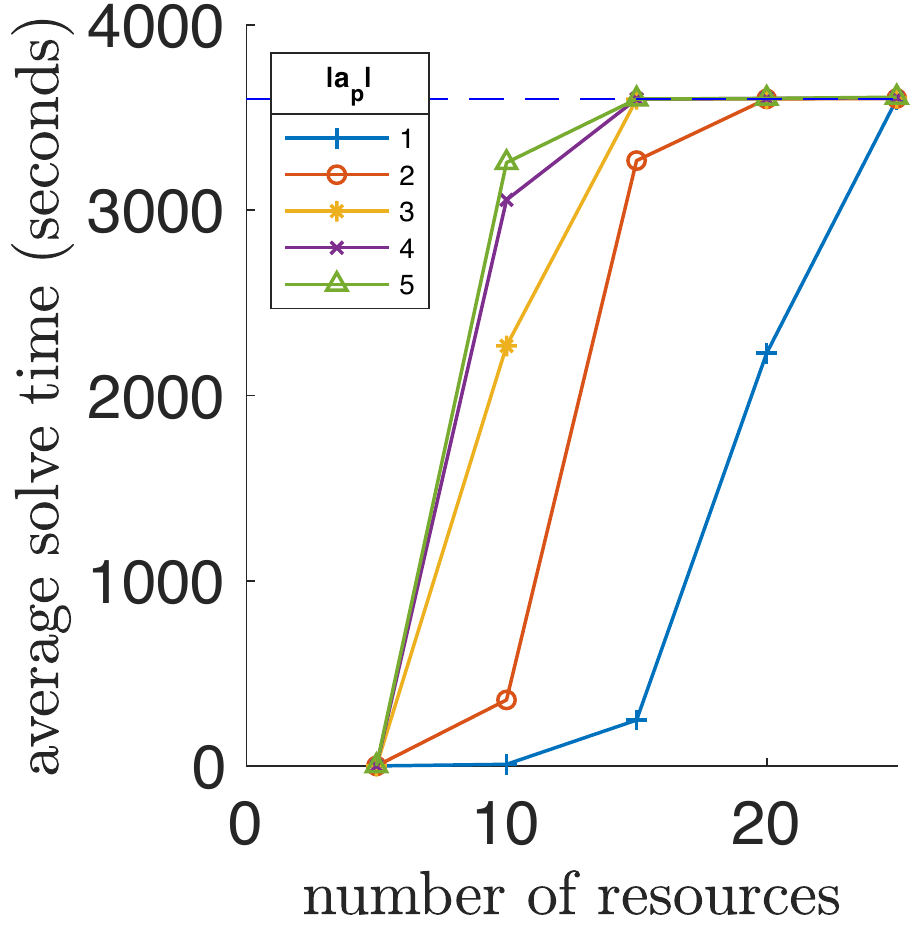}
		\subcaption{SCGs $(n = 2\,r)$}
%		\begin{center}
%			(e) SCGs $(n = 2\,r)$
%		\end{center}
	\end{minipage}
%	\begin{minipage}{4.3cm}
%		\includegraphics[scale=.44]{figures/figure3_non_sin.pdf}
%		\subcaption{SCGs $(n = 3\,r)$}
%		\begin{center}
%			(f) SCGs $(n = 3\,r)$
%		\end{center}
%	\end{minipage}
	\begin{minipage}{4.3cm}
		\includegraphics[scale=.38]{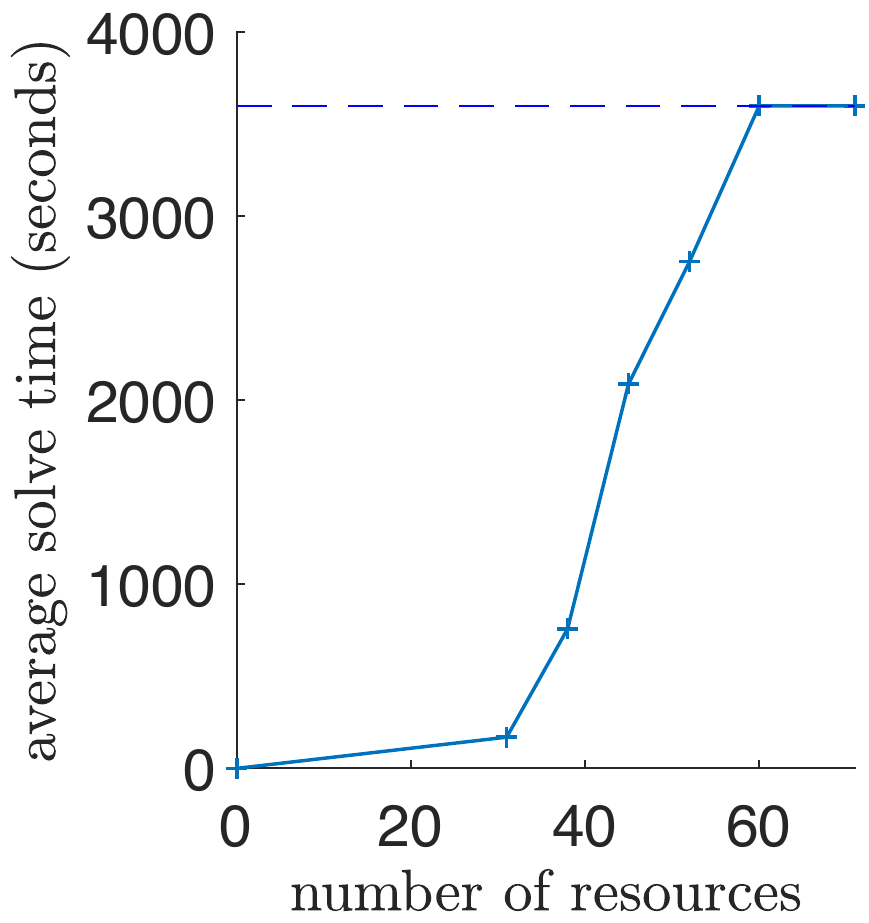}
		\subcaption{Worst-case SCGs}
%		\begin{center}
%			(f) SCGs $(n = 3\,r)$
%		\end{center}
	\end{minipage}	\caption{Computing times of Formulations~\eqref{eq:milp_dif_types}~and~\eqref{eq:milp_general} on randomly generated game instances and worst-case instances.}
% built on the base of our hardness reductions, as described in Section~\ref{sec:experiments}.}
	\label{fig:exp}
	%		\vspace{0.5mm}
\end{figure*}

\section{Experimental Results}\label{sec:experiments}

% In this section, we experimentally evaluate the MILP formulations proposed in Section~\ref{sec:formulations} on a testbed of randomly generated game instances. 
%
%In this section, 
%We experimentally evaluate the MILP formulations proposed in Section~\ref{sec:formulations}.
% 
%Let us remark that 
There is no standard testbed for our games and there is no evidence that some structured games are more representative than others.
Thus, we test the MILP formulations proposed in Section~\ref{sec:formulations} on randomly generated games, which represent instances of average-case complexity, and games based on the reductions provided in the proofs of Theorems~\ref{thm:hard_general_symm}~and~\ref{thm:np_hard_types}, which, instead, represent worst-case complexity instances.

All the experiments are run on a UNIX machine with a total of 32 cores working at 2.3 GHz, equipped with 128 GB of RAM.
Each instance is solved on a single core within a time limit of 3600 seconds.
We use GUROBI 7.0 (with Python interface) as MILP solver.

\subsubsection{Random Game Instances}
%

%{\bf Random Game Instances.}
%
For $\mathcal{T}$-class SSCGs, we generate random game instances with $r \in \{10,20,30,40,50,60,70,80,90,100\}$ resources and $T \in \{1,2,3,4\}$ classes, with $n_t \in \{0.2\,r,0.5\,r,r\}$ followers per class $t \in \mathcal{T}$ and $|A_t|=0.5\,r$ actions per class $t \in \mathcal{T}$.
Cost functions are randomly generated by sampling uniformly from $\{1,\ldots,n r T\}$. 
For general SCGs, we generate instances with $r \in \{5,10,15,20,25\}$ resources and $n \in \{r,2\,r,3\,r\}$ followers, with $|a_p| \in \{1,2,3,4,5\}$ resources per action $a_p\in A_p$ and $|A_p|=0.5\,r$ actions per player $p \in N$.
Cost functions are randomly generated by sampling uniformly from $\{1,\ldots,nr\}$.
We build a testbed with $20$ game instances per combination of the parameters.

% Figure~\ref{fig:exp} shows average computing times of our formulations. 
%
%%%%%Figures~\ref{fig:exp}a,~\ref{fig:exp}b,~and~\ref{fig:exp}c display the average computing times for %%%%%%Formulation~\eqref{eq:milp_dif_types} with $0.2\,r$, $0.5\,r$, and $r$ followers per class, %%%%%%respectively. 
Figure~\ref{fig:exp}(a) displays the average computing times for Formulation~\eqref{eq:milp_dif_types} with $0.5\,r$ followers per class. 
The formulation scales quite well in practice. 
Symmetric games ($T=1$) are quickly solved up to $r=100$.
Moreover, we are able to solve within the time limit games with up to four classes, $40$ resources, and $160$ players ($40$ players per class).
Let us notice that the dynamic programming algorithm presented in Theorem~\ref{thm:dyn} can be employed in this setting to find an OSE, if we restrict the leader to play pure strategies. However, preliminary experiments show that its scalability is extremely limited with respect to that of our formulation, as it finds a solution within the time limit only for games with less than 10 resources, while our formulation scales on much bigger games and it also works for mixed-strategy commitments.

% much easier to solve, while increasing the number of types games became harder.
%%%%%Figures~\ref{fig:exp}e,~\ref{fig:exp}f,~and~\ref{fig:exp}g show the average computing times for %%%%%Formulation~\eqref{eq:milp_general} with $r$, $2\,r$, and $3\,r$ followers, respectively.
Figure~\ref{fig:exp}(c) shows the average computing times for Formulation~\eqref{eq:milp_general} with $2\,r$ followers.
We can conclude that, as expected, game instances with non-singleton actions are much harder to solve than singleton games.
Here, the largest game instances we can solve within the time limit feature actions of cardinality two, $15$ resources, and $30$ players.~\footnote{We report additional experimental results in Appendix~\ref{sec:app_b}.}

\subsubsection{Worst-Case Instances}

%{\bf Worst-Case Instances.}
%
We test Formulation~\eqref{eq:milp_general} on instances built according to the reduction in Theorem~\ref{thm:hard_general_symm}.
Specifically, we generate these games from random \textsc{3SAT} instances with $|U| \in \{4,5,6,7,8,9\}$ variables and $|C| = k |U|$ clauses, where $k\approx4.26$ is the phase transition parameter determining generally hard-to-solve \textsc{3SAT} instances~\cite{phaseTransition}.
%
%Specifically, we generate these games from random 3SAT instances with $|C| \in \{2,4,6,8,10,12,14,16\}$ clauses and $|U| = 2|C|$ variables.
%
We test 10 random instances for each number of variables. 
Furthermore, we experiment Formulation~\eqref{eq:milp_dif_types} on instances based on the reduction in Theorem~\ref{thm:np_hard_types}.
We generate these games from random $K$-\textsc{PARTITION} instances with $|S| \in \{20, 40, \ldots,160 \}$ integers, $s_i \in [2,100]$ for all $s_i \in S$, and $K = \frac{|S|}{2}$.
%
%Even in this case, 
We test 10 random instances for each value of $|S|$.

Figures~\ref{fig:exp}(b)~and~\ref{fig:exp}(d) show the computing times for $\mathcal{T}$-class SSCGs and SCGs, respectively.
Surprisingly, the results we obtain are comparable to those for random games, empirically showing that, for the games we study, random instances are not easier to solve than structured ones, as instead it is the case, \emph{e.g.}, in normal-form games~\cite{DBLP:conf/aaai/SandholmGC05}.

%% file: conclusions.tex
\section{Discussion and Future Works}\label{sec:conclusions}

This paper studies the problem of computing Stackelberg equilibria in multi-follower Stackelberg games with a massive number of players, focusing on congestion games.
Our results shed light on the boundary between hard and easy game instances, significantly advancing the state of the art.

Our analysis about games with non-singleton actions shows that, in oder to compute an equilibrium in time polynomial in the number of players, singleton actions are necessary (Theorems~\ref{thm:hard_general_symm}~and~\ref{thm:inapx_general_symm}).
This answers a question left open by~\citeauthor{aij2018leadership}~\shortcite{aij2018leadership}.
Surprisingly, our negative result holds even in games where each action is made of at most two resources.
Future works could address non-singleton games with players' actions enjoying specific structures, as it is the case, \emph{e.g.}, in games played on graphs~\cite{werneck2000finding}.

Our findings about congestion games with multiple classes of players substantially improve known results on symmetric (\emph{i.e.}, single-class) games.
Our dynamic programming algorithm can compute optimal leader's pure-strategy commitments in time polynomial in the number of players, even when costs are arbitrary functions and the players are split into a (fixed) number of different classes (Theorem~\ref{thm:dyn}).

Future extensions of this work may investigate how adopting a different solution concept for the followers' game, such as, \emph{e.g.}, the \emph{correlated equilibrium}~\cite{aumann1974subjectivity}, affects the computation of equilibria.
%For instance, we may assume that the followers play a \emph{correlated equilibrium}~\cite{aumann1974subjectivity}.
%
This has been recently studied for polynomial-type games~\cite{ijcai19multi}.

%% file: appendix.tex
%\begin{table*}
%	\centering
%	\begin{tabular}{c}
%		\textbf{\LARGE Online Appendix for the IJCAI-19 Submission 5162}
%	\end{tabular}
%\end{table*}

\section{Omitted Proofs}\label{sec:app_a}

The proof of Theorem~\ref{thm:hard_general_symm} is based on reduction from \textsc{3SAT}, which reads as follows.

\begin{definition}[\textsc{3SAT}]\label{def:3sat}
	Given a finite set $C$ of 3-literal clauses defined over a finite set $U$ of variables, is there a truth assignment to the variables which satisfies all clauses? 
	\footnote{ We refer to a \textsc{3SAT} instance as $(C,U)$.
		Moreover, $l \in \phi$ denotes a literal (\emph{i.e.}, a variable or its negation) appearing in $\phi \in C$, while $u(l)$ is the variable corresponding to that literal.
		%
		%We also let $|C|=m$ and $|V|=s$ be, respectively, the number of clauses and variables.
	}
\end{definition}

\theoremone*
\begin{proof}
	We provide a reduction from \textsc{3SAT} showing that a polynomial-time algorithm for finding an OSE in SCGs would allow us to solve any \textsc{3SAT} instance in polynomial time.
	Given a \textsc{3SAT} instance $(C,U)$ and a number $0 < \epsilon < 1$, we build an SCG $\Gamma_\epsilon(C,U)$ admitting an OSE $\sigma$ with $c_\ell^\sigma = \epsilon$ if and only if $(C,U)$ is satisfiable.

	\emph{Mapping.}
	$\Gamma_\epsilon(C,U)$ is defined as follows:
	%
	%\vspace{-1.5mm}
	\begin{itemize}
		\item $N=F\cup \{\ell\}$ with $F = \{ p_u \mid u \in U \} \cup \{ p_\phi \mid \phi \in C \}$;
		\item $R = \{ r_w \} \cup \{ r_{u}, r_{\bar u} , r_{u,t} \mid u \in U \} \cup \{ r_\phi \mid \phi \in C \}$;
		\item $A_p \hspace{-0.5mm}=\hspace{-0.5mm}  \{ a_u \hspace{-0.5mm}= \hspace{-0.5mm} \{r_{u}, r_{u,t}\} ,a_{\bar u} \hspace{-0.5mm}=\hspace{-0.5mm} \{r_{\bar u}, r_{u,t}\} \mid u \in U \}  \cup  \{ a_w \hspace{-0.5mm}=\hspace{-0.5mm} \{r_w\} \} \cup \{ a_{\phi, l} =  \{r_\phi, r_l\} \mid \phi \in C, l \in \phi\} $ for all $p \in N$.
	\end{itemize}
	%\vspace{-1.5mm}
	%
	Cost functions are specified in the following table, and, additionally, $c_{r_{\bar u},f} = c_{r_{u},f}$ and $c_{i,\ell} = c_{i,f}$ for all $i \in R$.
	%
	%\vspace{-1.5mm}
	\begin{center}
		{\renewcommand{\arraystretch}{1.1}\begin{tabular}{c|cccc}
								\hline
								$x$ & $c_{r_\phi,f}$ & $c_{r_{u},f}$  & $c_{r_{u,t},f}$ &   $c_{r_{w},f}$ \\
								\cline{1-5}
								$1$ & $1$ & $0$ & $3$  & $\epsilon$ \\
								$ [ 2,|C|+|U|+1 ] $  & $5$ & $2$ & $5$ &  $4$ \\
							\end{tabular}}
	\end{center}
	%
%	\begin{table}[h]
%		\centering
%		{\renewcommand{\arraystretch}{1.3}\begin{tabular}{c|cccc}
%				\hline
%				$x$ & $c_{r_\phi,f}$ & $c_{r_{u},f}$  & $c_{r_{u,t},f}$ &   $c_{r_{w},f}$ \\
%				\cline{1-5}
%				$1$ & $1$ & $0$ & $3$  & $\epsilon$ \\
%				$ [ 2,|C|+|U|+1 ] $  & $5$ & $2$ & $5$ &  $4$ \\
%			\end{tabular}}
%			\caption{Cost functions of games $\Gamma_\epsilon(C,U)$ for the reduction in the proof of Theorem~\ref{thm:hard_general_symm}. The other cost functions are $c_{r_{\bar u},f} = c_{r_{u},f}$ and $c_{i,\ell} = c_{i,f}$ for all $i \in R$.}
%			\label{tab:costs_np_hard_opt}
%			%\vspace{-0.5cm}
%		\end{table}
		%
		%
		Clearly, given $(C,U)$, $\Gamma_\epsilon(C,U)$ can be constructed in polynomial time, since $n = |C| + |U| + 1$, $r = |C| + 3 |U| + 1$, and $|A_p| = 3 |C| + 2 |U| +1$ for all $p \in N$.
		Let us remark that $\Gamma_\epsilon(C,U)$ is symmetric, cost functions are monotonic, and each action has cardinality at most two.
		Moreover, the leader's cost is $\epsilon$ if and only if she is the only player who plays the singleton action $a_w$, otherwise her cost is at least $1$, since other actions contain two resources.
		
		\emph{If.} 
		Suppose that $(C,U)$ is satisfiable, and let $\tau : U \rightarrow \{\mathsf T,\mathsf F\}$ be a truth assignment satisfying all clauses in $C$. 
		Let $\sigma_\ell \in \Delta_\ell : \sigma_\ell(a_w)=1$.
		Using $\tau$, we can build $\sigma = (\sigma_\ell, a)$, with $a \in E^{\sigma_\ell}$, such that $c_\ell^\sigma = \epsilon$.
		Since $\epsilon$ is the minimum leader's cost and the followers behave optimistically, $\sigma$ is an OSE.
		In particular, for every $\phi \in C$, there must be a follower $p \in F$ such that $a_p = a_{\phi,l}$, where $l \in \phi$ evaluates to true under $\tau$. 
		%
		%$a$ must prescribe one follower to play an action $a_{\phi,l}$, where $l \in \phi$ is a literal that evaluates to true under $\tau$. 
		%
		Clearly, 
		%since $\tau$ satisfies all clauses, 
		one such literal $l \in \phi$ always exists.
		When there are many, 
		%the follower must select the best one, \emph{i.e.}, 
		take one minimizing $\nu_{r_l}^a$.
		Moreover, for every $u \in U$, if $\tau(u)=\mathsf T$, respectively $\tau(u)=\mathsf F$, there must be a follower $p \in F$ such that $a_p =a_{\bar{u}}$, respectively $a_p =a_{{u}}$.
		%
		%$a$ must prescribe one follower to play action $a_{\bar{v}} $, respectively $a_v$.
		%
		%Instead, if $\tau(v)=\mathsf F$, let a follower playes  action $\{p_{v},p_{v,t}\}$.
		%
		% It is easy to see that we can always construct $a = (a_p)_{p \in F} \in A_F$ satisfying the properties described above.
		%, since there are $m+s$ followers who can select any action.
		%
		Thus, $\nu_{r_\phi}^a = 1$ for all $\phi \in C$, and, similarly, $\nu_{r_{u,t}}^a = 1$ for all $u \in U$.
		Additionally, $\nu_{r_w}^a = 0$ as there are $|C|+|U|$ followers.
		Next, we show that $a \in E^{\sigma_\ell}$.
		First, followers $p \in F$ with $a_p = a_{\phi,l}$ experience a cost $c_p^\sigma = c_{r_\phi, f}(\nu_{r_\phi}^a) + c_{r_l, f}(\nu_{r_l}^a) \leq 3$, since $\nu_{r_\phi}^a = 1$.
		Thus, they do not have any incentive to deviate.
		%, for the following reasons.
		%
		If they switch to $a_{\phi', l'}$ (with $\phi' \neq \phi$), then they would pay at least $5$ (as $\nu_{r_{\phi'}}^a = 1$).
		Furthermore, they do not deviate to $a_{\phi, l'}$ (with $l' \neq l \in \phi$), as, if $l'$ is false, then they would pay $3$, while, when $l'$ is true, they would incur a cost of $c_{r_\phi, f}(\nu_{r_\phi}^a) + c_{r_{l'}, f}(\nu_{r_{l'}}^a) \geq  c_p^\sigma $ (as $\nu_{r_{l'}}^a \geq \nu_{r_{l}}^a$).
		%
		%greater than or equal to $c_p^\sigma$, given how $a$ is constructed.
		%
		If, instead, they deviate to $a_u$ or $a_{\bar u}$, then their cost would be at least $5$ (as $\nu_{r_{u,t}}^a = 1$).
		Moreover, they do not switch to $a_w$, since they would pay $4$.
		Followers $p \in F$ with $a_p =a_{{u}}$ has cost $c_p^\sigma = c_{r_{u},f}(\nu_{r_{u}}^a) + c_{r_{u,t},f}(\nu_{r_{u,t}}^a) = 3$ since $\nu_{r_{u}}^a = 1$.
		Thus, they do not deviate, as they would pay at least $4$.
		Similarly, followers $p \in F$ with $a_p =a_{\bar{u}}$ do not deviate.
		As a result, $a$ is an NE and, since $\sigma_\ell(a_w)=1$ and $\nu_{r_w}^a = 0$, it holds $c_\ell^\sigma=\epsilon$.

		\emph{Only if.}
		Suppose there exists an OSE $\sigma = (\sigma_\ell,a)$ such that $c_\ell^\sigma = \epsilon$.
		%
		%We show that $a = (a_p)_{p \in F} \ A_p$ can be employed to recover, in polynomial %time, a truth assignment $\tau$ that satisfies all clauses in $C$.
		%
		Thus, $\sigma_\ell(a_w) = 1$ and $\nu_{r_w}^a = 0$.
		%First, let us notice that, since the leader's cost is $\epsilon$, $\sigma_\ell(a_w) = 1$ and $\nu_{r_w}^a = 0$.
		%
		For $u \in U$, 
		%it must be the case that 
		$\nu_{r_{u,t}}^a \leq 1$, otherwise, if $\nu_{r_{u,t}}^a\geq 2$, some followers would have an incentive to deviate to action $a_w$, paying $4 < 5$.
		%
		%at most one follower can select one between $a_v$ and $a_{\bar v}$, otherwise $\nu_{r_{v,t}}^a\geq 2$ and the followers would have an incentive to deviate to $a_t$ (as they would pay $4<5$).
		%
		Analogously, for $\phi \in C$, 
		%it must be the case that 
		$\nu_{r_{\phi}}^a \leq 1$.
		%
		% at most one follower can select one among the actions $\{ a_{\phi, l} \mid l \in \phi \}$, otherwise $\nu_{r_{\phi}}^a\geq 2$ and the followers would have an incentive to switch to $a_t$ (as $4<5$).
		%
		%As a result, s
		Since there are $|C|+|U|$ followers, $\nu_{r_{u,t}}^a = 1$ for every $u \in U$, and $\nu_{r_{\phi}}^a = 1$ for every $\phi \in C$.
		%
		% exactly one follower selects one between $a_v$ and $a_{\bar v}$, and, for every clause $\phi \in C$, exactly one follower selects one among the actions $\{ a_{\phi, l} \mid l \in \phi \}$.
		%
		Thus, for every $u \in U$, there exists $p \in F$ such that either $a_p = a_{ u}$ or $a_p= a_{\bar u}$, and no other follower selects actions $a_{u}$ and $a_{\bar u}$.
		Define a truth assignment $\tau$ such that $\tau(u)=\mathsf T$ if there is $p \in F$ with $a_p = a_{\bar u}$, while $\tau(u)= \mathsf F$ if there is $p \in F$ with $a_p = a_u$.
		Clearly, $\tau$ is well-defined. 
		Moreover, for every $\phi \in C$, there exists a unique follower $p \in F$ and a literal $l \in \phi$ such that $a_p = a_{\phi, l}$, as $\nu_{r_{\phi}}^a = 1$.
		This implies that no follower plays $a_l$, otherwise her cost would be at least $5$, and she would deviate to $a_w$, paying $4$.
		Thus, if $l$ is positive, there is $p \in F$ with $a_p = a_{\bar u}$, while, if it is negative, there is $p \in F$ with $a_p = a_ u$.
		%
		% as previously shown, for every $\phi \in C$, there exists a literal $l \in \phi$ such that a follower plays $a_{\phi, l}$, which implies that no follower plays $a_l$ (otherwise the follower would pay at least $5$ and she would have an incentive to deviate to $a_t$), and, thus, a follower plays $a_{\bar v}$, and $\tau(v(l))=\mathsf T$ if $l$ is positive, while $\tau(v(l))=\mathsf F$ if it is negative.
		%
		Therefore, $\tau$ satisfies all clauses.
\end{proof}

\theoremtwo*

\begin{proof}
	Given a \textsc{3SAT} instance $(C,U)$, we build an SCG $\Gamma_\epsilon(C,U)$ as in the proof of  Theorem~\ref{thm:hard_general_symm}.
	As previously shown, in an OSE $\sigma$ of $\Gamma_\epsilon(C,U)$, it holds $c_\ell^\sigma = \epsilon$ if and only if $(C,U)$ is satisfiable.
	Next, we prove that, if $(C,U)$ is \emph{not} satisfiable, then any OSE $\sigma$ has $c_\ell^\sigma \geq 1$.
	Suppose, by contradiction, there exists an OSE $\sigma = (\sigma_\ell, a)$ with $c_\ell^\sigma < 1$. 
	This implies that $\sigma_\ell(a_w) > 0$ and $\nu_{r_w}^a = 0$, otherwise $c_\ell^\sigma \geq 1$. 
	Moreover, all the followers must experience a cost at most of $4$, otherwise they would have an incentive to switch to $a_w$.
	Thus, for every $u \in U$, it must be the case that $\nu_{r_{u,t}}^a \leq 1$, otherwise, if $\nu_{r_{u,t}}^a\geq 2$, some followers would have a cost at least $5$.
	Similarly, for every $\phi \in C$, it must be the case that $\nu_{r_{\phi}}^a \leq 1$.
	Following the same reasoning as in the proof of Theorem~\ref{thm:hard_general_symm}, we can build a truth assignment satisfying all clauses, a contradiction.
	%
	%We show that, as a consequence, there exists an OSE in which leader pays only $epsilon$. 
	%
	%If leader plays $\sigma_\ell(a_t)=1$, leader cost decrease to $\epsilon$ and also the costs of all followers don't increase (costs are monotonic). Since followers cost are $\leq 4$, for each variable $v \in V$, at most one follower selects $a_v$ or $a_{\bar v}$, since otherwise $\nu_{r_{v,t}}^a\geq 2$ and followers costs would be $\geq 5$, and, for each clause $\phi \in C$, at most one follower selects one among the actions $\{ a_{\phi, l} \mid l \in \phi \}$, since otherwise $\nu_{r_{\phi}}^a\geq 2$ and followers costs would be $\geq 5$.
	%As we have show in the \textbf{Only if.} of Theorem~\ref{thm:hard_general_symm}, this implies that this is an OSE with leader cost $\epsilon$ and $(C,V)$ is satisfiable, contradicting our assumption.
	%
	Finally, let $\epsilon = \frac{1}{2^{I}}$, where $I$ is the size of $\Gamma_\epsilon(C,U)$.
	% (polynomiality of the reduction is preserved).
	%
	Suppose there is a polynomial-time approximation algorithm $\mathcal{A}$ 
	%for the problem of computing an OSE in SCGs 
	with approximation factor $\text{poly}(I)$, \emph{i.e.}, a polynomial function of $I$.
	If $(C,U)$ is satisfiable, then $\mathcal{A}$ applied to $\Gamma_\epsilon(C,U)$ would return a solution with cost at most $\frac{1}{2^I} \text{poly}(I) < 1$, for $I$ large enough.
	Thus, $\mathcal{A}$ would allow us to solve any \textsc{3SAT} instance in polynomial time, which is a contradiction unless \textsf{P} = \textsf{NP} holds.
\end{proof}

\theoremthree*

\begin{proof}
	The result is readily obtained from the proofs of Theorems~\ref{thm:hard_general_symm}~and~\ref{thm:inapx_general_symm}, by setting $A_\ell = \{a_w\}$ in the reduction.
	%
	%In order to prove the result, 
	%It is sufficient to change the reduction in the proof of Theorem~\ref{thm:hard_general_symm}, setting $A_\ell = \{a_w\}$, while all the rest remains the same.
	%
	%Since, in these new games, the leader has to play $a_w$, the proof is like those of Theorems~\ref{thm:hard_general_symm}~and~\ref{thm:inapx_general_symm}.
	%
\end{proof}

\lemmaone*

\begin{proof}
	%Constraint~\eqref{dp:1} allows to quantify how many players, out of $B$, choose resource $h$ and, therefore, the number of those who choose resources in $\{1, \dots, h-1\}$.
	%
	We show that all the constraints are necessary.
	% for the definition of $\mathcal{O}(i,B_1,\dots,B_T,M_1,\dots,M_T,V_1,\dots,V_T)$.
	%
	%Let us show that the %% other
	%constraints are necessary for the definition of $\mathcal{O}(i,B_1,\dots,B_T,,M_1,\dots,M_T,V_1,\dots,V_T)$ to be respected.
	%
	If Constrains~\eqref{dp:15} were not satisfied, at least one player would play an action not available to her.
	If Constraints~\eqref{dp:14} were not satisfied, there would exist a $t \in \mathcal{T}$ such that $m_t > M_t$, and, thus, there would be at least a resource in $\{1, \dots, i-1\}$ having cost larger than $M_t$ for players of class $t$.
	If Constraints~\eqref{dp:13} were not satisfied, there would exist a $t \in \mathcal{T}$ such that $v_t < V_t$, and, thus, players of class $t$ would incur a cost strictly smaller than $V_t$ when deviating to a resource in $\{1, \dots, i-1\}$.
	%
	%the cost when deviating to a resource among those in $\{1, \dots, i-1\}$ is strictly smaller than $V_t$ for players of type $t$.
	%
	If Constraints~\eqref{dp:2} were not satisfied, there would exist a $t \in \mathcal{T} :B_t - p_t > 0$ such that $c_i(b) > M_t$, and, thus, $M_t$ would be smaller than the cost of the most expensive resource used by players of class $t$.
	If Constraints~\eqref{dp:4} were not satisfied, there would exist a $t \in \mathcal{T}: i \in A_t$ such that $c_i(b+1) < V_t$, and, thus, players of class $t$ would incur a cost strictly smaller than $V_t$ upon deviating to another resource.
	If Constraints~\eqref{dp:3} were not satisfied, there would exist a $t \in \mathcal{T} : B_t -p_t>0$ such that $c_i(b) > v_t$ and, thus, at least one player of class $t$ using resource $i$ would have an incentive to deviate to another resource.
	If Constraints~\eqref{dp:5} were not satisfied, there would exist a $t \in \mathcal{T} : i \in A_t$ such that $c_i(b+1) < m_t$ and at least one player of class $t$ experiencing a cost of $m_t$ would prefer to switch to resource $i$.
\end{proof}

\theoremfour*

\begin{proof}
	Since there are at most $n r$ different values of costs $c_i(x)$ ($i \in R$, $x \in \{1,\ldots,n\}$), for each $t \in \mathcal{T}$ there are at most $nr$ values of $M_t$ and $V_t$.
	There are also exactly $r$ values of $i \in R$ and exactly $n_t$ values of $B_t$ for each $t \in \mathcal{T}$.
	%, with $\sum_{t \in \mathcal{T}} n_t = n-1$.
	%
	Hence, the dynamic programming table of $\mathcal{O}(i,B_1,\dots,B_T,M_1,\dots,M_T,V_1,\dots,V_T)$ contains $O(n^{3T} r^{2T+1})$ entries.
	Computing an entry of the table requires $O(n^{3T} r^{2T})$.
	Overall, an optimal NE is computed in $O(n^{6T} r^{4T+1})$.
	For $\mathcal{T}$-class SSCGs with the leader restricted to play pure strategies, it suffices to run the algorithm for each $i \in A_\ell$, using as optimality criterion the minimization of the cost of $i$.
	Since there are $O(r)$ leader's actions, the overall complexity is $O(n^{6T} r^{4T+2})$.
\end{proof}

The proof of Theorem~\ref{thm:np_hard_types} is based on a reduction from $K$-\textsc{PARTITION}, which reads as follows.

\begin{definition}[$K$-\textsc{PARTITION}]
	Given a finite set $S = \{ s_1, \ldots, s_{|S|} \}$ of positive integers $s_i \in \mathbb{N}^+$ with $\sum_{s_i \in S} s_i = 2X$ and a positive integer $K \leq \frac{|S|}{2}$, is there a subset $S' \subseteq S$ such that $|S'|=K$ and $\sum_{s_i \in S'} s_i = X$?
\end{definition}

\thmpart*
\begin{proof}
	%	We provide a reduction of $K$-PARTITION, showing that any polynomial-time algorithm for finding an OSE in $\mathcal{T}$-type SSCGs would allow us to solve any $K$-PARTITION instance in polynomial time.
	%	%
	%	Given a $K$-PARTITION instance $(S,K)$, we build a game $\Gamma(S,K)$ in which there is an OSE providing the leader with a cost less than or equal to $2X - \frac{X}{K}$ if and only if $(S,K)$ admits a \emph{yes} answer, i.e., there exists $S' \subseteq S$ with $|S'|=K$ and $\sum_{s_i \in S'} s_i = X$.
	%	%
	%	W.l.o.g., let $s_i \leq X$ for all $s_i \in S$ (if not, clearly $(S,K)$ has answer \emph{no}).
	Our reduction from $K$-\textsc{PARTITION} shows that a polynomial-time algorithm for finding an OSE in $\mathcal{T}$-class SSCGs would allow us to solve any $K$-\textsc{PARTITION} instance in polynomial time.
	Given a $K$-\textsc{PARTITION} instance $(S,K)$, we build a game $\Gamma(S,K)$ that admits an OSE $\sigma$ with $c_\ell^\sigma \leq 2X - \frac{X}{K}$ if and only if $(S,K)$ has answer \emph{yes}, \emph{i.e.}, there is $S' \subseteq S$ with $|S'|=K$ and $\sum_{s_i \in S'} s_i = X$.
	Let, w.l.o.g., $s_i \leq X$ for all $s_i \in S$ (if not, $(S,K)$ has answer \emph{no}).
	
	{\em Mapping.}
	$\Gamma(S,K)$ is defined as follows:
	%
	%\vspace{-1.5mm}
	\begin{itemize}
		\item $N = F \cup \{ \ell \}$ and $\mathcal{T} = \{ 1,2,3,4 \}$, where $F = \bigcup_{t \in \mathcal{T}} F_t$ with $|F_1|= K$, $|F_2|=2|S|$, $|F_3|= 1$, and $|F_4|=1$;
		\item $R = R_S \cup \{ r_w, r_x, r_y, r_z \}$ with $R_S= \{ r_i \mid s_i \in S \}$;
		\item $A_1 = R_S \cup \{ r_w \}$, $A_2 = R_S \cup \{ r_z \}$, $A_3 = \{ r_w, r_y \}$, $A_4 = \{ r_x, r_y \}$, and $A_\ell = R_S \cup \{ r_y \}$.
	\end{itemize}
	%\vspace{-1.5mm}
	%
	Costs are specified in the table below, with $C_{r_y, f}=\frac{6K-2}{2K^2 - K}$, $ C_{r_i, f}= \left(1 - \frac{2XK}{s_i} + 2XK \right) \frac{2XK}{s_i}$, and $C_{r_i, \ell}= \frac{2X (2X - s_i)}{s_i}$.
	%
	%\vspace{-1.5mm}
	\begin{center}
		{\renewcommand{\arraystretch}{1.3}
			\setlength{\tabcolsep}{4pt}
			\begin{tabular}{c|ccccccc}
				\hline
				$x$ & $c_{r_i,f}$ & $c_{r_w,f}$ & $c_{r_x,f}$  & $c_{r_y,f}$ & $c_{r_z,f}$ & $c_{r_i,\ell}$ & $c_{r_y,\ell}$ \\
				\cline{1-8}
				$1$ & 0 & $\frac{1}{K}$ & $\frac{3}{K}$ & $\frac{2}{K}$ & $2XK$& $C_{r_i, \ell}$& 0  \\
				$2$ & $\frac{2XK}{s_i}$ & 1 & $\frac{3}{K}$ & $C_{r_y, f}$ & $2XK$& $C_{r_i, \ell}$&$X^4$  \\
				$[3, n]$ & $C_{r_i, f}$ & 1 & $\frac{3}{K}$ & $C_{r_y, f}$& $2XK$ & $X^4$&$X^4$  \\
		\end{tabular}}
	\end{center}
	%
	%
	%	\begin{table}[h]
	%		\centering
	%		{\renewcommand{\arraystretch}{1.3}
	%			\setlength{\tabcolsep}{5pt}
	%			\begin{tabular}{c|ccccccc}
	%			\hline
	%			$x$ & $c_{r_i,f}$ & $c_{r_w,f}$ & $c_{r_x,f}$  & $c_{r_y,f}$ & $c_{r_z,f}$ & $c_{r_i,\ell}$ & $c_{r_y,\ell}$ \\
	%			\cline{1-8}
	%			$1$ & 0 & $\frac{1}{|S|}$ & $\frac{3}{|S|}$ & $\frac{2}{|S|}$ & $XK$& $C_{r_i, \ell}$& 0  \\
	%			$2$ & $\frac{XK}{x_i}$ & 1 & $\frac{3}{|S|}$ & $C_{r_y, f}$ & $XK$& $C_{r_i, \ell}$&$X^4$  \\
	%			$[3, n]$ & $C_{r_i, f}$ & 1 & $\frac{3}{|S|}$ & $C_{r_y, f}$& $XK$ & $X^4$&$X^4$  \\
	%		\end{tabular}}
	%		\caption{Cost functions of $\Gamma(S,K)$ for the proof of Theorem~\ref{thm:np_hard_types}, where $ C_{r_i, f}= \left(1 - \frac{XK}{x_i} + XK \right) \frac{XK}{x_i}$, $C_{r_y, f}=\frac{6|S|-2}{2|S|^2 - |S|}$, and $C_{r_i, \ell}= \frac{X (2X - x_i)}{x_i}$.}
	%		\label{tab:costs_np_hard_types}
	%		%\vspace{-0.5cm}
	%	\end{table}
	%
	Clearly, $\Gamma(S,K)$ can be constructed in polynomial time, since $n = K + 2 |S| + 3$, $r = |S| + 4$, $|A_1|=|A_2|=|A_\ell|= |S|+1$, $|A_3|=|A_4|=2$, and each cost can be encoded with a number of bits polynomial in the size of the instance $(S,K)$.
	Notice that resource costs are monotonic in the congestion.
	
	{\em If.} Suppose that $(S,K)$ has answer \emph{yes}, and let $S' \subseteq S$ be such that $|S'|=K$ and $\sum_{s_i \in S'} s_i = X$.
	Using $S'$, we can recover $\sigma = (\sigma_\ell, \{ \nu_t \}_{t \in \mathcal{T}})$ such that the followers' configurations $\{ \nu_t \}_{t \in \mathcal{T}}$ represent an NE for $\sigma_\ell$ and $c_\ell^\sigma = 2X - \frac{X}{K}$.
	Thus, in any OSE the leader's cost must be less than or equal to $2X - \frac{X}{K}$.
	In particular, for every $s_i \in S'$, let $\nu_{1, r_i} = 1$, $\nu_{2, r_i} = 0$, and $\sigma_\ell(r_i) = \frac{s_i}{2XK}$.
	Instead, for $s_i \notin S'$, let $\nu_{1, r_i} = 0$, $\nu_{2, r_i} = 2$, and $\sigma_\ell(r_i) =0$.
	Moreover, we let $\nu_{1, r_w} = 0$, $\nu_{2, r_z} = 2K$, $\nu_{3, r_w} = 1$, $\nu_{3, r_y} = 0$, $\nu_{4, r_x} = 1$, $\nu_{4, r_y} = 0$, and $\sigma_\ell(r_y) =\frac{2K-1}{2K}$.
	It is easy to see that both $\{ \nu_t \}_{t \in \mathcal{T}}$ and $\sigma_\ell$ are well-defined.
	Next, we prove that $\{ \nu_t \}_{t \in \mathcal{T}}$ represent an NE for $\sigma_\ell$.
	First, followers of class 1 experience a cost of $\frac{2XK}{s_i} \frac{s_i}{2XK} = 1$, and, thus, they do not have incentive to deviate to resource $r_w$, as they would still pay $1$.
	Similarly, they do not switch to another resource $r_i \in R_S$, since, if $s_i \in S'$, they would get a cost of $\frac{2XK}{s_i} (1 - \frac{s_i}{2XK}) + C_{r_i,f} \frac{s_i}{2XK} = 2XK> 1$, while, if $s_i \notin S'$, they would pay $C_{r_i,f} > 1$.
	Moreover, followers of class 2 do not deviate, since, if they are selecting a resource $r_i \in R_S$, then their current cost is $\frac{2XK}{s_i}$ and they would pay at least $2XK \geq \frac{2XK}{s_i}$ by deviating, while, if they are using $r_z$, then they experience a cost of $2XK$ and they would pay at least $2XK$ by switching to a resource $r_i \in R_S$.
	Finally, the follower of class 3 pays $\frac{1}{K}$ and she does not deviate to $r_y$, as she would incur a cost at least of $\frac{2}{K}$, and, analogously, the follower of class 4 does not deviate since her cost would be $C_{r_y, f} \frac{2K-1}{2K} +  \frac{2}{K} \frac{1}{2K}= \frac{3}{K}$ and she is paying $\frac{3}{K}$.
	In conclusion, $c_\ell^\sigma = \sum_{s_i \in S'} C_{r_i,\ell} \frac{s_i}{2XK} = \sum_{s_i \in S'} \frac{2X}{K} - \sum_{s_i \in S'} \frac{s_i}{K} = 2X - \frac{X}{K}$, as $|S'|=K$ and $\sum_{s_i \in S'} s_i = X$.
	
	{\em Only if.} Suppose there exists an OSE $\sigma = (\sigma_\ell, \{ \nu_t \}_{t \in \mathcal{T}})$ such that $c_\ell^\sigma \leq 2X - \frac{X}{K}$.
	Using $\sigma$, we build $S' \subseteq S$ such that $|S'|=K$ and $\sum_{s_i \in S'} s_i = X$, showing that $(S,K)$ has answer \emph{yes}.
	First, it must be the case that $\sigma_\ell(r_y) > 0$ and $\nu_{r_y} = 0$, otherwise the leader's cost cannot be smaller than the minimum among costs $C_{r_i,\ell}$, which, since $s_i \leq X$ for all $s_i \in S$, is at least $\frac{2X(2X - X)}{X} = 2X > 2X - \frac{X}{K}$.
	Thus, it must be $\nu_{3, r_y} = \nu_{4, r_y} = 0$ and $\nu_{3, r_w} = \nu_{4, r_x} = 1$.
	As a result, $\sigma_\ell(r_y) \geq 1 - \frac{1}{2K}$, otherwise the follower of class 4 would have an incentive to deviate to resource $r_y$, paying $C_{r_y, f} \sigma_\ell(r_y) +  \frac{2}{K} (1-\sigma_\ell(r_y)) < \frac{3}{K}$.
	This implies that $\sum_{r_i \in R_S} \sigma_\ell(r_i) \leq \frac{1}{2K}$.
	Moreover, $\nu_{1,r_w} = 0$, otherwise, if $\nu_{1,r_w} > 0$, the follower of class 3 would experience a cost of $1$ and she would switch to $r_y$, paying at most $C_{r_y,f} < 1$, assuming, w.l.o.g., $K \geq 4$. 
	Thus, there must be $K$ different resources $r_i \in R_S$ such that $\nu_{1, r_i} = 1$ and $\nu_{2, r_i} = 0$, since, if either $\nu_{1, r_i} > 1$ or $\nu_{2, r_i} > 0$, then $\nu_{r_i} > 1$ and the followers of class 1 selecting $r_i$ would experience a cost greater than or equal to $\frac{2XK}{s_i} > 1$, thus having an incentive to deviate to $r_w$, paying $1$.
	Let $R_S' = \{ r_i \in R_S \mid \nu_{1,r_i}=1 \}$ (notice that $|R_S'|=K$).
	It must be the case that $\nu_{2,r_i} < 3$ for all $r_i \notin R_S'$, otherwise a follower of class 2 would have an incentive to deviate to $r_z$ (as $C_{r_i,f} > 2XK$).
	Thus, since $|F_2|=2|S|$, there are at least $2K$ followers on $r_z$.
	Furthermore, $\nu_{2,r_i} > 0$ for all $r_i \notin R_S'$, otherwise the followers of class 2 selecting $r_z$ would have an incentive to switch to $r_i$, paying less than $\frac{2XK}{s_i} \leq 2XK$.
	Now, let us fix any $r_i \in R_S'$.
	Say $\sigma_\ell(r_i) < \frac{s_i}{2XK}$, then the followers of class 2 using $r_z$ would deviate to $r_i$, paying $\frac{2XK}{s_i} (1 - \sigma_\ell(r_i)) + C_{r_i,f} \sigma_\ell(r_i) < 2XK$.
	Moreover, say $\sigma_\ell(r_i) > \frac{s_i}{2XK}$, then the follower of class 1 on $r_i$ would deviate to $r_w$, paying $1 < \frac{2XK}{s_i} \sigma_\ell(r_i)$.
	As a result, $\sigma_\ell(r_i) = \frac{s_i}{2XK}$ for every $r_i \in R_S'$.
	Since $\sum_{r_i \in R_S} \sigma_\ell(r_i) \leq \frac{1}{2K}$, we also have $\sum_{r_i \in R_S'} \sigma_\ell(r_i) = \frac{1}{2XK} \sum_{r_i \in R_S'} s_i \leq \frac{1}{2K}$, implying $ \sum_{r_i \in R_S'} s_i \leq X$.
	It must also be the case that $\sigma_\ell(r_i) = 0$ for all $r_i \notin R_S'$.
	If not, then there would be $r_j \in R_S'$ with $\nu_{2,r_j} \in \{1,2\}$ and $\sigma_\ell(r_j) >0$, which implies that $c_\ell^\sigma = X^4 \sigma_\ell(r_j) + \sum_{r_i \neq r_j \in R_S}  C_{r_i,\ell} \sigma_\ell(r_i) > \sum_{r_i \in R_S'} \frac{2X - s_i}{K} = 2X -  \frac{1}{K} \sum_{r_i   \in R_S'} s_i \geq 2X - \frac{X}{K} $, a contradiction.
	Finally, $c_\ell = \sum_{r_i \in R_S'} C_{r_i,\ell} \sigma_\ell(r_i) = 2X - \frac{1}{K} \sum_{r_i \in R_S'} s_i \leq 2X - \frac{X}{K}$, which implies $ \sum_{r_i \in R_S'} s_i \geq X$.
	Thus, $ \sum_{r_i \in R_S'} s_i = X$.
	Letting $S' = \{s_i \in S \mid r_i \in R_S' \}$, we have $|S'|=K$ and $ \sum_{s_i \in S'} s_i = X$.
	%We can also prove that $\nu_{2,r_i} =2 $ for all $r_i \notin R_S'$.
	%
	%If not, then either $\nu_{2,r_i}=0$ and followers of type 2 selecting $r_z$ would have an incentive to deviate to $r_i$ paying less than $\frac{2XK}{s_i} \leq 2XK$, or $\nu_{2,r_i}=1$ and the follower of type 2 using $r_i$ would prefer switching to a resource in $R_S'$.
	%
\end{proof}

\section{Additional Experimental Results}\label{sec:app_b}

Figure~\ref{fig:exp_full} reports all the running times of our formulations on the game instances composing our testbed.

Moreover, Tables~\ref{tab:types}~and~\ref{tab:non_sing} show average optimality gaps for the game instances where a feasible solution is found within the time limit.
The optimality gap is computed as $\frac{UB- LB}{UB}$, where $UB$ is the value of the best solution found by the solver within the time limit and $LB$ is the lower-bound returned by the solver.
We only provide the results for the largest game instances.
Specifically, Table~\ref{tab:types} reports the results for Formulation~\eqref{eq:milp_dif_types} on random $\mathcal{T}$-class SSCGs with $n_t = r$, while Table~\ref{tab:non_sing} shows the results for Formulation~\eqref{eq:milp_general} on random SCGs with $n = 3\,r$.

\begin{table}[!ht]
%\begin{small}
	\centering
	{\renewcommand{\arraystretch}{1.1}\setlength{\tabcolsep}{2pt}\begin{tabular}{r|c|c|c|c|c|}
		$T \, \backslash \, r$ & 10 & 30 & 50 & 70 & 90 \\ \hline
		1 & 0 (100) & 0.00 (100) & 0.00 (100) & 0.08 (100) & 0.35 (100) \\ \hline
		2 & 0 (100) & 0.00 (100) & 0.41 (100) & 0.51 (100) & 0.81 (70) \\ \hline
		3 & 0 (100) & 0.17 (100) & 0.79 (100) & 0.93 (100) & 0.95 (70) \\ \hline
		4 & 0 (100) & 0.43 (100) & 0.90 (90) & 0.97 (45) & 0.99 (10) \\ \hline
	\end{tabular}}
	%\end{small}
	%
	\caption{Avg. optimality gaps (\% of games where feasible solution found) of Form.~\eqref{eq:milp_dif_types} on random $\mathcal{T}$-class SSCGs ($n_t = r$).}
	\label{tab:types}
\end{table}
\begin{table}[!ht]
%\begin{scriptsize}
	\centering
	{\renewcommand{\arraystretch}{1.1}\setlength{\tabcolsep}{2pt}\begin{tabular}{r|c|c|c|c|c|}
		$|a_p| \, \backslash \, r$ & 5 & 10 & 15 & 20 & 25 \\ \hline
		1 & 0 (100) & 0.00 (100) & 0.15 (100) & 0.63 (100) & 0.63 (100) \\ \hline
		2 & 0 (100) & 0.25 (100) & 0.90 (100) & 0.95 (95) & 0.98 (40) \\ \hline
		3 & 0 (100) & 0.76 (100) & 0.98 (80) & 0.99 (10) & 0.99 (10) \\ \hline
		4 & 0 (100) & 0.82 (100) & 0.99 (35) & -- (0) & -- (0) \\ \hline
		5 & 0 (100) & 0.88 (100) & 0.99 (5) & -- (0) & -- (0) \\ \hline
	\end{tabular}}
%\end{scriptsize}
	%
	\caption{Avg. optimality gaps (\% of games where feasible solution found) of Formulation~\eqref{eq:milp_general} on random SCGs ($n = 3r$).}
	\label{tab:non_sing}
\end{table}

\begin{figure*}[!htp]
	%\vspace{1mm}
	\begin{minipage}{4.3cm}
		\includegraphics[scale=.44]{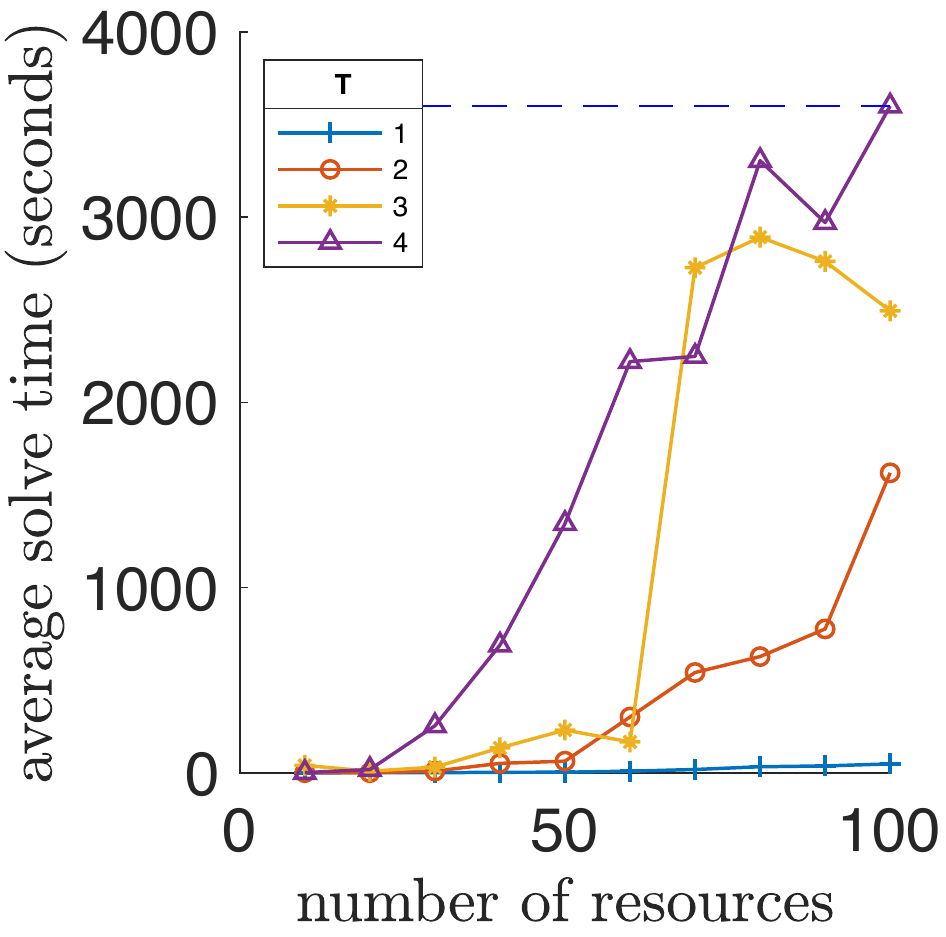}
		\subcaption{$\mathcal{T}$-class SSCGs $(n_t = 0.2\,r)$}
		%		\begin{center}
		%			(d) SCGs $(n = r)$
		%		\end{center}
	\end{minipage}
	\begin{minipage}{4.3cm}
		\includegraphics[scale=.44]{figure2_types.pdf}
		\subcaption{$\mathcal{T}$-class SSCGs $(n_t = 0.5\,r)$}
		%		\begin{center}
		%			(e) SCGs $(n = 2\,r)$
		%		\end{center}
	\end{minipage}
	\begin{minipage}{4.3cm}
		\includegraphics[scale=.44]{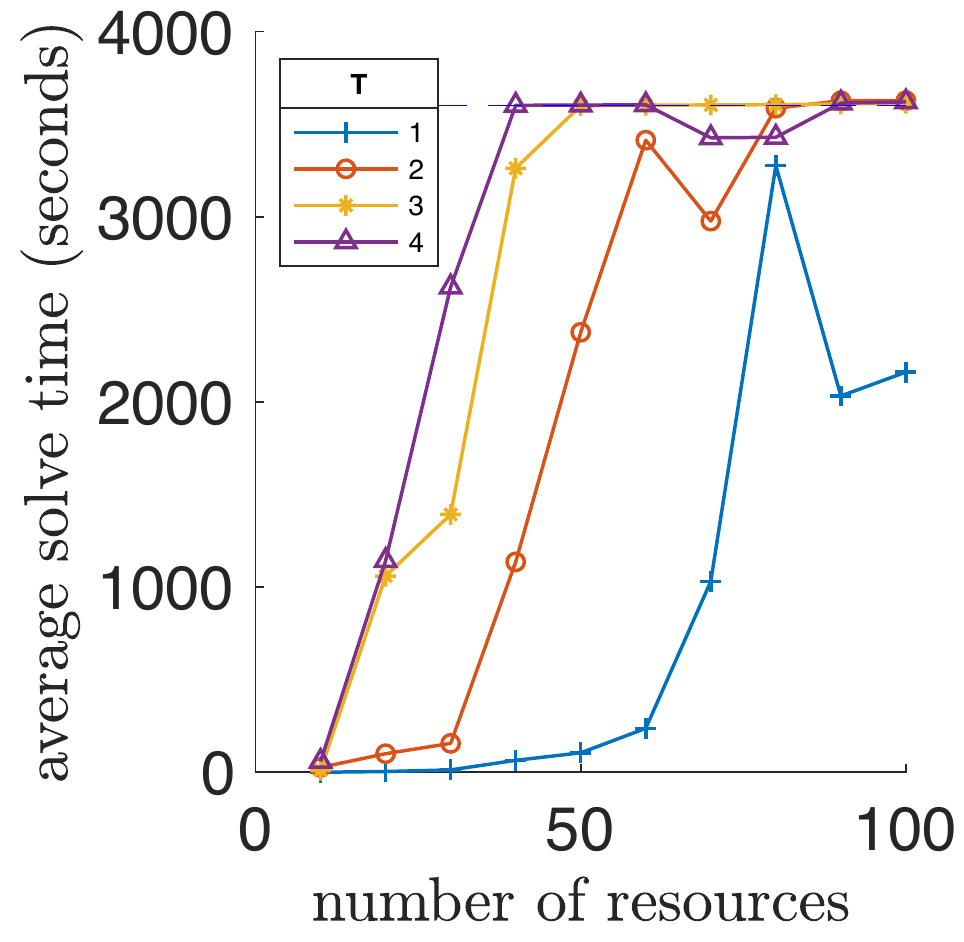}
		\subcaption{$\mathcal{T}$-class SSCGs $(n_t = r)$}
		%		\begin{center}
		%			(f) SCGs $(n = 3\,r)$
		%		\end{center}
	\end{minipage}
	\begin{minipage}{4.3cm}
		\includegraphics[scale=.44]{figPartition.pdf}
		\subcaption{Worst-case $\mathcal{T}$-class SSCGs}
		%		\begin{center}
		%			(f) SCGs $(n = 3\,r)$
		%		\end{center}
	\end{minipage}
	\begin{minipage}{4.3cm}
		\includegraphics[scale=.44]{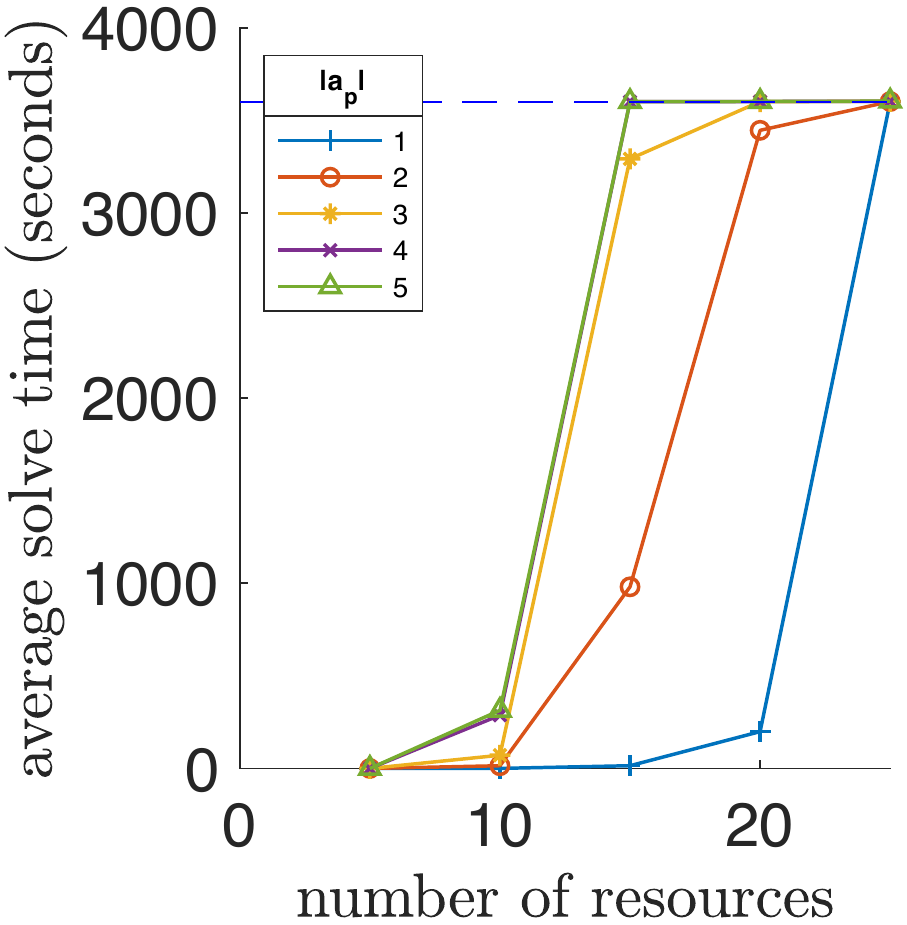}
		\subcaption{ SCGs $(n = r)$}
		%		\begin{center}
		%			(d) SCGs $(n = r)$
		%		\end{center}
	\end{minipage}
	\begin{minipage}{4.3cm}
		\includegraphics[scale=.44]{figure2_non_sin.pdf}
		\subcaption{SCGs $(n = 2\,r)$}
		%		\begin{center}
		%			(e) SCGs $(n = 2\,r)$
		%		\end{center}
	\end{minipage}
	\begin{minipage}{4.3cm}
		\includegraphics[scale=.44]{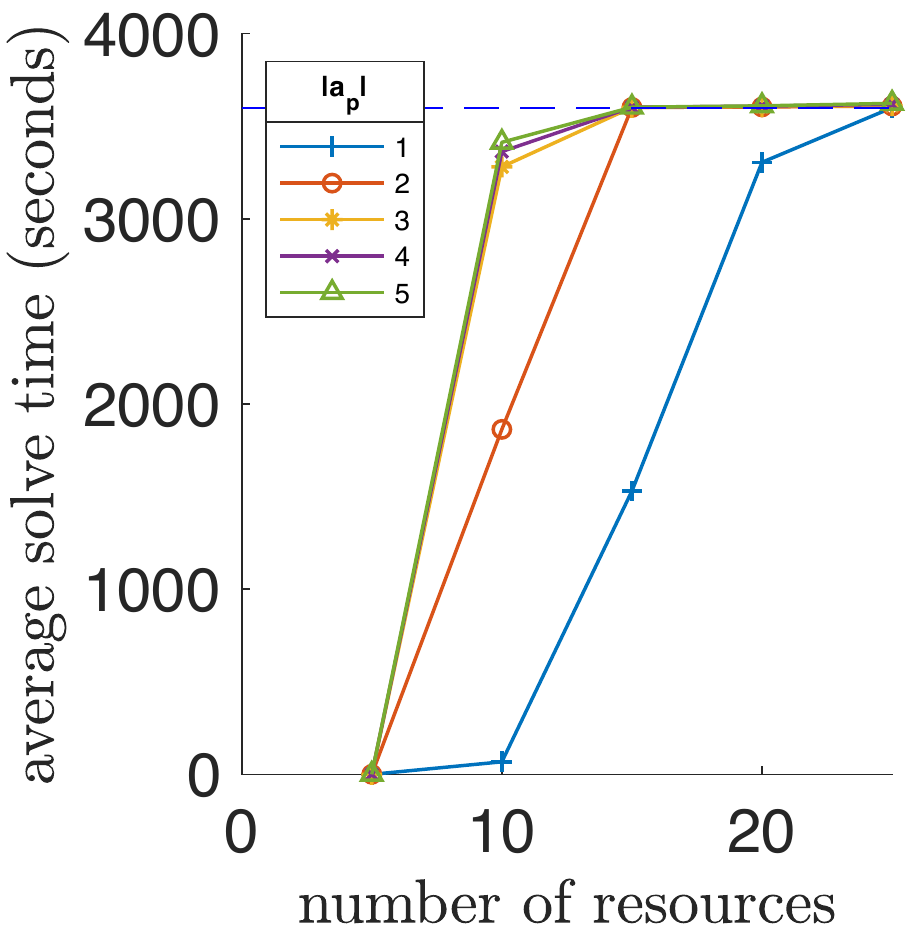}
		\subcaption{SCGs $(n = 3\,r)$}
		%		\begin{center}
		%			(f) SCGs $(n = 3\,r)$
		%		\end{center}
	\end{minipage}
	\begin{minipage}{4.3cm}
		\includegraphics[scale=.44]{fig3sat.pdf}
		\subcaption{Worst-case SCGs}
		%		\begin{center}
		%			(f) SCGs $(n = 3\,r)$
		%		\end{center}
	\end{minipage}	\caption{Computing times of Formulations~\eqref{eq:milp_dif_types}~and~\eqref{eq:milp_general} on randomly generated game instances and worst-case instances built on the base of our hardness reductions, as described in Section~\ref{sec:experiments}.}
	\label{fig:exp_full}
	%		\vspace{0.5mm}
\end{figure*}